\documentclass[pdflatex,sn-aps,Numbered]{sn-jnl}

\usepackage{graphicx}%
\usepackage{multirow}%
\usepackage{amsmath,amssymb,amsfonts}%
\usepackage{amsthm}%
\usepackage{mathrsfs}%
\usepackage[title]{appendix}%
\usepackage{xcolor}%
\usepackage{textcomp}%
\usepackage{manyfoot}%
\usepackage{booktabs}%
\usepackage{algorithm}%
\usepackage{algorithmicx}%
\usepackage{algpseudocode}%
\usepackage{listings}%
\usepackage{cuted}
\usepackage{mathtools}
\usepackage{makecell}
\usepackage{subcaption}

\theoremstyle{thmstyleone}%
\newtheorem{theorem}{Theorem}
\newtheorem{proposition}[theorem]{Proposition}%

\theoremstyle{thmstyletwo}%

\theoremstyle{thmstylethree}%
\newtheorem{definition}{Definition}%

\begin{document}

\title[]{Why dimensional analysis works: general classification of self-similarity based on scale-invariance}


\author*[1]{\fnm{Hirokazu} \sur{Maruoka}}\email{hirokazu.maruoka@oist.jp,hmaruoka1987@gmail.com}

\affil*[1]{\orgdiv{Nonlinear and Non-equilibrium Physics Unit}, \orgname{Okinawa Institute of Science and Technology (OIST)}, \orgaddress{\street{Tancha, Onna-son, Kunigami-gun} \city{Okinawa}, \postcode{904-0495},  \country{Japan}}}


\abstract{In this work, we formulate self-similarity from the perspective of scale invariance, where a self-similar form is understood as the transformation of a function into a form invariant under scale transformations. By applying this formulation to physical parameters, which consist of numerical values and units, it is demonstrated that dimensional analysis works for physical problems because scale invariance is partially shared between units and physical parameters. This naturally leads to the distinction between similarity of the first kind and similarity of the second kind according to whether the scale functions induced by units and those associated with physical parameters are equivalent or not. Self-similar solutions of the second kind can be further classified according to whether the power exponents of the similarity parameters include functions of dimensionless numbers. This leads to the conclusion that there are three kinds of self-similar solutions. The present work provides a unified framework for understanding dimensional analysis and a universal classification of self-similarity in physical problems. }

\maketitle

\section{Introduction  \label{sec:Introduction} }

Self-similarity is a feature that does not change under scale transformations. It is a fundamental concept in physics
\cite{Galilei1638,Maxwell1873,Rayleigh1915,Mandelbrot1983,Barenblatt1996, Barenblatt2003, Goldenfeld1992}.
Every physical system involves self-similarity \cite{Baumchen2013,Benzaquen2013,BoakyeAnsah2021,Zhou2026}, as long as it involves
physical quantities. Self-similarity appears as scale invariance. Data points collapse onto lower-dimensional manifolds by introducing similarity parameters, which is called data collapse \cite{Cabella2011,Yokota2011,Okumura2020,Dresselhaus2022,Watanabe25}. This concept leads to dimensional analysis \cite{Bertrand1878,Riabouchinsky1911,Federman1911,Buckingham1914,Gibbings2011}, which
allows one to obtain solutions without directly solving the problems \cite{Taylor1950, Diaz2021}.
Where there is self-similarity, there are scaling laws or power laws.
They bridge theory and experimental observations, and provide
useful information for engineering science. Scaling laws appear in many
fields of science, including condensed matter \cite{Widom1963, Holmes1971, Ley-Koo1977},
fluid mechanics \cite{deGennes2004,Barenblatt2014,Mahesh2017}, soft matter
\cite{Flory1953, deGennes1979, Maruoka2025b}, contact mechanics \cite{Maruoka23,Yokoyama2026}, allometry \cite{West1997}, fractal theory \cite{ElNabulsi2023,Samy2025,Samy2026},
and machine learning \cite{kaplan2020}.

It was later recognized that there are two kinds of self-similarity in terms of the convergence of similarity parameters \cite{Guderley1942, Weizsacker1954}. Some problems reveal divergence in the limit of the similarity parameters,
but convergence is recovered by introducing new similarity variables with anomalous dimensions. Zel'dovich formalized such problems as self-similarity of the second kind \cite{Zeldovich1956}, and this concept was later developed through the introduction of intermediate asymptotics \cite{Zeldovich1972}.

In more detail, according to the recipe of Barenblatt \cite{Barenblatt1996,Barenblatt2003}, similarity of the first kind and similarity of the second kind can be distinguished by the convergence of the self-similar form in the limit of the parameters. Suppose that there is a physical function, $y = f\left(t,x,z \right)$. By applying dimensional analysis, a dimensionless function is obtained as $\Pi = \Phi\left(\zeta, \eta \right)$, where $\Pi = y/t^\alpha$, $\zeta = x/t^\beta$, and $\eta = z/t^\gamma$. If $\Phi$ converges to a finite limit as $\eta$ and $\zeta$ go to zero or infinity, the problem turns out to be a similarity of the first kind. On the other hand, if $\Phi$ does not converge to a finite limit as $\eta$ and $\zeta$ go to zero or infinity, but the convergence is recovered by introducing a similarity variable $\zeta / \eta^{\delta}$, it turns out to be a similarity of the second kind. According to this classification, it seems that the similarity of the first kind or the second kind depends on the scale invariance of the similarity parameters obtained by dimensional analysis.

On the other hand, Eggers and Fontelos proposed a slightly different formulation \cite{Eggers2015}. They classified self-similarity according to whether there exists a self-similar solution for a particular pair of values of the power exponents or there exist self-similar solutions for a continuous range of exponents to be determined by regularity conditions. In this formulation, dimensional analysis does not play a central role; rather, whether the power exponents characterizing the self-similar parameters are fixed or not determines the kind of self-similarity. However, this formulation clearly casts light on the point that the power exponents are determined through regularity conditions in the second kind.

These two formulations definitely overlap, though they involve some ambiguity. This suggests that another category of self-similar solutions is demanded. In this paper, I formulate self-similarity based on scale invariance of {\it physical parameters} to integrate all these formulations. It will be demonstrated that the self-similarity of the first kind and the second kind are actually derived from the structure of physical parameters, which has not been sufficiently appreciated. It can be represented as {\it products of numerical values and units}. It means that there are two kinds of scale transformations, which naturally lead to the distinction between similarity of the first kind and similarity of the second kind. Self-similarity of the second kind can be defined as problems for which the scale invariance of units and that of physical parameters are not equal. This structure also naturally leads to the {\it hierarchical structure} of self-similarity. We will see that similarity of the second kind bifurcates according to whether the power exponents of the similarity parameters constitute functions of dimensionless numbers or not. Summarizing these formulations, we will see that there are three kinds of self-similar solutions.

This insight also addresses a fundamental question of self-similarity: {\it why} and {\it how} does dimensional analysis actually work for physical problems despite the fact that it only considers units? Units are more artificial properties, introduced for the purpose of measurement and their choice should have no influence on the physical phenomena being observed. Nevertheless, dimensional analysis provides physical insight into the observed phenomena from the structure of units. Existing treatments and books explain dimensional analysis in terms of dimensional homogeneity or scale invariance, but provide little discussion of why information about physical scaling can be inferred from units alone \cite{Bertrand1878,Riabouchinsky1911,Federman1911,Buckingham1914,Bridgman1922,Ipsen1960,Olver1986,Barenblatt1996,Barenblatt2003,Gibbings2011,Hehl2019}. This is counter-intuitive. However, it is easily understood by carefully observing the structure of physical parameters, which consists of numerical values and units. The reason why dimensional analysis is effective for physical problems is derived from the indistinguishability of the action of scale transformations on numerical values. This work addresses this question. 

The purpose of this paper is to clarify why dimensional analysis works for physical systems and to classify self-similarity from this viewpoint. The rest of this paper is organized as follows. In Sec. \ref{sec:SII}, we first show three cases of self-similar solutions that reveal different types of self-similarity. In Sec. \ref{sec:SIII}, we formulate scale transformations and the exploitation of self-similarity in terms of scale invariance. In Sec. \ref{sec:SIV}, we introduce the form of physical parameters. The introduction of physical parameters leads to two types of scale transformations acting on physical parameters, which leads to the distinction between similarity of the first kind and similarity of the second kind. In Sec. \ref{sec:SV}, we further formulate self-similarity of the second kind depending on whether the scale functions include dimensionless variables. In Sec. \ref{sec:SVI}, we summarize the formulation of self-similarity. In this paper, we use the term self-similarity in a broader sense as invariance under scale transformations, not necessarily restricted to asymptotic solutions of PDEs.

\section{Three kinds of self-similar solutions\label{sec:SII}}

Here we would like to show examples of three kinds of self-similar solutions, which will be classified later. Any self-similar solution can be transformed into a form invariant under scale transformations through a transformation of parameters. We refer to this invariant form as a {\it self-similar form}. The introduced parameters in self-similar forms called {\it similarity parameters} or {\it similarity variables}. The function appearing in a self-similar form is called a {\it scaling function}. By examining these self-similar forms, we easily find that there are three kinds of self-similar solutions based on the forms of the similarity parameters.

---Case I. Suppose that we have the following diffusion equation \cite{Atkins2006},
\begin{equation}
\partial_t u = D \partial_{xx} u \label{eq:e1a}
\end{equation}
with the boundary condition $x\in(-\infty, \infty)$ and the initial condition $u\left(0,x \right) = Q \delta\left( x \right)$. Its solution can be expressed as
\begin{equation}
u = f \left(t, x, D, Q \right)\label{eq:e2a}
\end{equation}
where $u$ is the concentration of the substance, $t$ is time, $x$ is the spatial coordinate, $D$ is the diffusion coefficient, and $Q$ is the total amount. Its self-similar form can be obtained by simply applying dimensional analysis. In $LMT$ units, the dimensions of the physical parameters are $\left[ u \right] = L^{-1}M$, $[t] = T$, $[x] = L$, $[D] = L^{2}T^{-1}$, and $[Q] = M$. Therefore, we obtain the following self-similar form,
\begin{equation}
\frac{ u \sqrt{Dt}}{Q} = f \left( \frac{x}{\sqrt{Dt}}\right). \label{eq:e3a}
\end{equation}
The data points collapse by plotting the variables in Eq.~\ref{eq:e3a}, which shows that all scale invariance has been exploited in this expression, as shown in Fig.~\ref{fig:F1}. Note that the self-similar form is obtained simply by applying dimensional analysis.

\begin{figure}
\centering
  \includegraphics[width=8.5 cm]{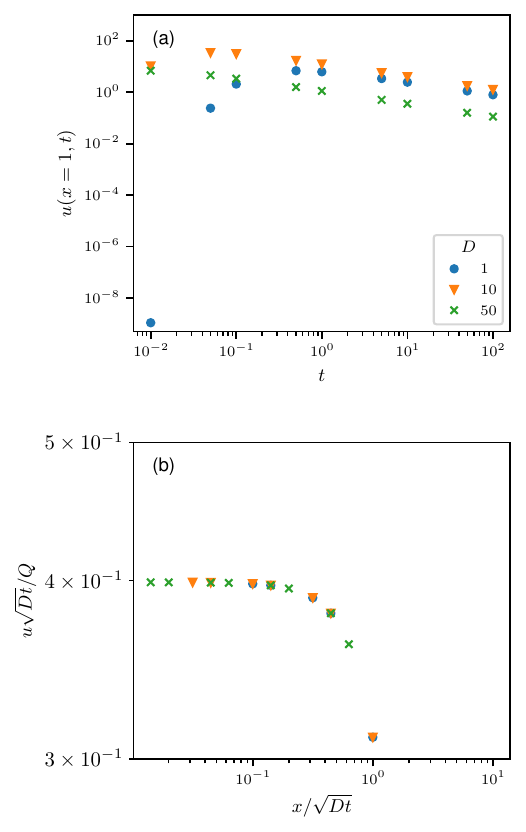} %
 \caption{Plots of a one-dimensional diffusion process for various diffusion coefficients, $D = 1, 10, 50$:
(a) the concentration of the substance at $x=1$ as a function of time $t$;
(b) data collapse obtained by plotting the similarity parameters, $\frac{u \sqrt{Dt}}{Q}$ versus $\frac{x}{\sqrt{Dt}}$.
 } 
  \label{fig:F1}
\end{figure}

---Case II. Let us suppose the dynamical impact of a solid sphere onto a viscoelastic board \cite{Maruoka23}. In this problem, the physical parameters are expressed as
\begin{equation}
\delta_m = f \left( R, h, \rho, \mu, E, v_i   \right) \label{eq:e4a}
\end{equation}
where $\delta_m$ is the maximum deformation, $R$ is the radius of the sphere, $h$ is the thickness of the board, $\rho$ is the density of the sphere, $\mu$ and $E$ are the viscous coefficient and the elastic modulus of the viscoelastic board, respectively, and $v_i$ is the impact speed. The dimensions of the physical parameters are $[\delta_m] = L$, $[R]= L$, $[h] = L$, $[\rho] = L^{-3} M$, $[\mu]=L^{-1} M T^{-1}$, $[E]= L^{-1}MT^{-2}$, and $[v_i]= LT^{-1}$.

In this problem, by applying dimensional analysis, we obtain
\begin{equation}
\frac{\delta_m}{R} =f \left( \frac{h}{R}, \frac{\mu}{E^{1/2} \rho^{1/2 }R}, \frac{\rho v_i^2}{E}  \right). \label{eq:e5a}
\end{equation}
However, the data points do not collapse by using the variables in Eq.~\ref{eq:e5a}. This is because the dimensional analysis alone does not reduce the problem to a relation between two similarity parameters. Actually, the self-similar form of this problem is
\begin{equation}
\frac{\Pi^3}{\kappa \eta} = \Phi \left( \frac{\Pi}{\theta \eta^{1/2}  }\right) \label{eq:e6a}
\end{equation}
where $\Pi = \frac{\delta_m}{R}$, $\kappa= \frac{h}{R}$, $\theta = \frac{\mu}{E^{1/2} \rho^{1/2 }R}$, and $\eta =\frac{\rho v_i^2}{E}$. This self-similar form was obtained by considering its physical model, the Maxwell viscoelastic foundation model \cite{Maruoka23}. Note that the self-similar form was not obtained by dimensional analysis. However, the power exponents of the similarity variables are constant. This is the characteristic feature of this case.

---Case III. The modified groundwater flow problem \cite{Barenblatt2003}. Groundwater filters through porous media with fissurized rock, where water is absorbed into the rock. Under this condition, the groundwater behavior is described by the following modified porous medium equation,
\begin{equation}
\partial_t H = \kappa \partial_{xx} H^2 - 2c \kappa \left(\partial_x H \right)^2
\end{equation}
where $H$ is the groundwater head, $\kappa$ is a constant consisting of the permeability coefficient, the viscous coefficient, and the porosity, and $c$ is a dimensionless parameter consisting of the porosity and the fissure porosity. The solution is expressed as
\begin{equation}
H = f \left(t, x, \kappa, I, l, c \right).
\end{equation}
where $l$ is the initial section, and $I$ is the total head defined as $I = \int_{-l}^l H(x,t)dx$. The dimensions of the physical parameters are $
\left[H \right] = \frac{M}{LT^2}$,~$\left[ t\right] = T$,~$\left[x \right]=L,$ $\left[ \kappa \right]= \frac{L^3T}{M} $~$\left[I \right]=\frac{M}{T^2}$,
$\left[l \right]=L$, $[c]=1 $.  
Applying dimensional analysis, we obtain
\begin{equation}
    \Pi = f\left(\zeta, \eta, c \right) \label{eq:e10a}
\end{equation}
where
\begin{equation}
\Pi = \frac{H \left(\kappa t\right)^{1/3}}{I^{2/3}}, \qquad
\zeta =\frac{x}{\left(I \kappa t \right)^{1/3}}, \qquad
\eta = \frac{l}{\left(I \kappa t \right)^{1/3}}. \nonumber
\end{equation}
In this problem, the self-similar form cannot be obtained simply by applying dimensional analysis, as in Case II. However, the final self-similar form is expressed as follows:
\begin{equation}
\frac{\Pi}{\eta^{\frac{2c}{3-2c} }} = \Phi \left( \frac{\zeta}{\eta^{\frac{c}{3-2c}}} \right),
\end{equation}
which was obtained by solving nonlinear eigenvalue problem \cite{Barenblatt2003}.
Note that the power exponents of the variables constitute functions of the dimensionless parameter $c$. 

Comparing these three self-similar forms, it is clear that there are three kinds of self-similar solutions. However, this classification has not been well founded. According to the formulation of Barenblatt, in which the distinction based on whether the self-similar forms are obtained by dimensional analysis or not is emphasized \cite{Barenblatt1996,Barenblatt2003,Barenblatt2014}, Case I is clearly a self-similar solution of the first kind and Case III is introduced as a self-similar solution of the second kind, but Case II should also be a self-similar solution of the second kind. On the other hand, following the formulation of Eggers and Fontelos, in which solutions are distinguished according to whether the power exponents are constant or take values in a continuous range \cite{Eggers2015}, Case II should be a self-similar solution of the first kind.

This apparent confusion originates from different perspectives on self-similarity. In the next section, we formulate the classification based on scale invariance.

\section{Self-similarity as scale-invariance \label{sec:SIII} }

In this section, I formulate self-similarity in terms of the scale invariance of a system. We will see that only two theorems are sufficient to formalize self-similarity. The formulation presented here is based on Barenblatt \cite{Barenblatt1996,Barenblatt2003}. However, instead of restricting the analysis to dimensions and units, we generalize it to arbitrary physical parameters endowed with scale transformations. This section clarifies what the transformation into a self-similar form means. We present the essential theorems, while their proofs are given in the appendices.

\subsection{Scale transformation and scale functions}

Let us suppose that there is a following implicit function,
\begin{equation}
    F\left(x_1,\cdots, x_N \right)=0, \label{eq:e1}
\end{equation}
where $x_1, \cdots, x_N$ are parameters. 

A scale transformation is a transformation that enlarges (increases) or shrinks (diminishes) parameters by scale factors that are the same in all directions. It is always expressed by the following theorem:
\begin{theorem}\label{thm:-1}
The scale functions of continuous scale transformations are always expressed as power-law monomials of the scaling factors:
\begin{eqnarray}
x_i \mapsto x'_i &=& \sigma_i\left(A_1, \cdots, A_k \right) x_i \nonumber \\
&=& A_1^{\omega_{1i}}\cdots A_k^{\omega_{ki}} x_i 
\label{eq:e2}
\end{eqnarray}
where $A_1,\cdots, A_k$ are {\it scaling factors} or simply factors, the exponents $\omega_{1i},\cdots,\omega_{ki}$ are {\it scaling exponents}, and $\sigma_i(A_1,\cdots,A_k)$ is the function determining the factor acting on $x_i$ ($i=1,\cdots, N$).
\end{theorem}
We call $\sigma_i$ a {\it scale function} or {\it scale transformation}. The important point is that scale functions are always power-law monomials. This is proved in Appendix \ref{sec:appendix_A}. 

When the scale functions act on the variable $x_1,\cdots,x_N$, the solutions are different in the scale preserving its similarity, 
\begin{equation}
       F\left(\sigma_1 x_1,\cdots,\sigma_N x_N \right)=0. \label{eq:e1a} 
\end{equation}
The solutions that satisfy Eq.~\ref{eq:e1a} is called {\it scale-invariant} \cite{Olver1986}. \footnote{Scale transformation can be understood as the transformation of a geometric object to its similar one. For instance, a circle can be expressed by $F\left(S,r \right)=S-\pi r^2=0$ where $S$ is a surface area of circle, $\pi$ is the ratio of circumference and $r$ is a radius. If the radius is transformed as $r \mapsto r^\prime = Ar$, then its surface area must increase with the factor $A^2$ as $S \mapsto S^\prime = A^2 S$ for the transformation to be a scale transformation that maps the circle onto a similar circle. If not, it will not maintain the identity of a circle. Note that it satisfies $F\left(Ar, A^2S \right)=\left(A^2S -\pi A^2 r^2 \right)= A^2(S-\pi r^2)=0$.}

\subsection{Functionally independent scale functions}

After defining scale transformations, we consider the transformation of a function into its self-similar form, which is invariant under its scale transformation. In order to understand this transformation, we introduce the notion of functional independence and a relevant theorem.

Now we define functionally independent scale functions. Suppose that there are scale functions $\sigma_1,\cdots, \sigma_N$ for parameters $x_1, \cdots, x_N$. If a scale function $\sigma_i$ cannot be expressed as a function of the other scale functions, $\sigma\left(\sigma_1, \cdots,\sigma_{i-1},\sigma_{i+1},\cdots, \sigma_N \right)$, then it is called a {\it functionally independent scale function}. On the other hand, if it can be expressed as a function of the other scale functions, it is called a {\it functionally dependent scale function}\footnote{Functional dependence generally implies that a function can be expressed as a function of other functions. It is easy to see by comparing functionally dependent functions, $x/y$ and $xy/(x^2+y^2)$. It can expressed as $xy/(x^2+y^2)=x/y/(1+(xy)^2))=f(x/y)$. On the other hand, $x/y$ and $x+y$ are functionally independent as $x+y=(x/y+1)/y=f(x/y,y)$. These two functions cannot be expressed as a function of function. See pp. 84--87 in Ref.~\cite{Olver1986}.}.

In $LMT$ units, the dimension functions, which determine the factors acting on the parameters under scale transformations of units, for force $F$, density $\rho$, and velocity $v$ are $\phi_F = LMT^{-2}$, $\phi_\rho = L^{-3}M$, and $\phi_{v}=LT^{-1}$. If they are functionally dependent, one of them should be expressible as a scale function of the others,
\begin{eqnarray}
\phi_F &=& \phi_\rho^x \phi_v^y  \nonumber \\
\iff LMT^{-2} &=& \left(L^{-3} M \right)^x \left( LT^{-1}\right)^y. \nonumber
\end{eqnarray}
Thus, there should be the following algebraic equations,
\begin{eqnarray}
L&:& -3x +y =1 \nonumber \\
M&:& x =1 \nonumber \\
T&:& y=2. \nonumber
\end{eqnarray}
However, there is no solution that satisfies these equations. Therefore, the dimension functions $\phi_F$, $\phi_\rho$, and $\phi_v$ are functionally independent scale functions.

After defining functionally independent scale functions, we arrive at the following theorem.

\begin{theorem}\label{thm:2-3}
A scale transformation is always possible such that any quantity, say $x_1$, in a set of parameters with functionally independent scale functions $x_1,\cdots, x_k$ changes its numerical value by a specified factor $A$ while the other quantities remain unchanged.
\end{theorem}
The proof of this theorem is given in Appendix \ref{sec:appendix_B}. This theorem implies that there exist scale transformations that change parameters associated with functionally independent scale functions while the other parameters do not change. It means that any function can be transformed into a form in which parameters with functionally dependent scale functions do not change, while parameters with functionally independent scale functions change.

\subsection{Exploitation of scale-invariance}\label{subsec:D}

By using the fact that scale functions are power-law monomials (Theorem \ref{thm:-1}) and Theorem \ref{thm:2-3}, we can deduce the transformation of a function into a form invariant under its scale transformations.

Let us consider an implicit function with two kinds of governing parameters, $x_1,\cdots, x_k$ and $y_1,\cdots, y_m$, 
\begin{equation}
    F\left(x_1,\cdots,x_k, y_1, \cdots, y_m \right) = 0, \label{eq:e18}
\end{equation}
where these parameters have the scale functions,
\begin{equation}
    F\left(\sigma_1 x_1,\cdots,\sigma_k x_k, \varsigma_1 y_1, \cdots, \varsigma_m y_m \right)=0, \label{eq:e19}
\end{equation}
where 
\begin{equation}
    \sigma_1=A_1^{\omega_{11}} \cdots A_k^{\omega_{k1}},\sigma_k=A_1^{\omega_{1k}} \cdots A_k^{\omega_{kk}},\varsigma_1 =A_1^{\gamma_{11}} \cdots A_k^{\gamma_{k1}},\cdots,\varsigma_m=A_1^{\gamma_{1m}} \cdots A_k^{\gamma_{km}}. \label{eq:e19a}
\end{equation}

Here there are two kind of parameters: 
\begin{itemize}
    \item[I] $x_1,\cdots,x_k$ : parameters with functionally independent scale functions, $\sigma_1,\cdots,\sigma_k$.
    \item[II] $y_1, \cdots, y_m$ : parameters with functionally dependent scale functions, $\varsigma_1,\cdots,\varsigma_m$.
\end{itemize}

Now let us rewrite the function in Eq.~\ref{eq:e18} in a form that remains unchanged, or {\it invariant}, under scale transformations.

As $\varsigma_1,\cdots,\varsigma_m$ are functionally dependent, we obtain,
\begin{eqnarray}
 \varsigma_1 &=& \sigma_1^{\alpha_{11}}\cdots\sigma_k^{\alpha_{k1}} \nonumber \\
 \varsigma_2 &=& \sigma_1^{\alpha_{12}}\cdots\sigma_k^{\alpha_{k2}} \nonumber \\
 \vdots \nonumber \\
 \varsigma_m &=& \sigma_1^{\alpha_{1m}}\cdots\sigma_k^{\alpha_{km}}. \label{eq:e20}
\end{eqnarray}
Here we introduce the following parameters,
\begin{equation}
 \Pi_1 = \frac{y_1}{x_1^{\alpha_{11}}\cdots x_k^{\alpha_{k1}}},\cdots,\Pi_m = \frac{y_m}{x_1^{\alpha_{1m}}\cdots x_k^{\alpha_{km}}}. \label{eq:e21}
\end{equation}
Indeed, we reformulate the function of Eq.~\ref{eq:e18},
\begin{eqnarray}
 F\left(x_1,\cdots,x_k, \Pi_1 x_1^{\alpha_{11}}\cdots x_k^{\alpha_{k1}},\cdots, \Pi_m x_1^{\alpha_{1m}}\cdots x_k^{\alpha_{km}}\right) \nonumber \\
 \iff F\left(x_1,\cdots,x_k, \Pi_1,\cdots, \Pi_m \right). \label{eq:e22}
\end{eqnarray}
Now, according to Theorem \ref{thm:2-3}, there exists a scale transformation such that $x_1,\cdots, x_k$, which belong to the set of parameters with functionally independent scale functions, change their numerical values by arbitrary factors while the other quantities remain unchanged. Therefore, we obtain,
\begin{equation}
 F\left(x_1,\cdots,x_k, \Pi_1,\cdots, \Pi_m \right) = \Phi\left(\Pi_1,\cdots, \Pi_m \right) = 0. \label{eq:e23}
\end{equation}
\begin{proposition}\label{prop:2-4}
The reformulated function $\Phi$ in Eq.~\ref{eq:e23} remains unchanged, or invariant, under scale transformations.
\end{proposition}

\begin{proof}
    
Now Eq.~\ref{eq:e23} does not change under the scale transformations. Using Eqs.~\ref{eq:e21}, we have
 \begin{eqnarray}
     \Phi\left(\Pi_1,\cdots, \Pi_m \right) = \Phi\left( \frac{y_1}{x_1^{\alpha_{11}}\cdots x_k^{\alpha_{k1}}}, \cdots, \frac{y_m}{x_1^{\alpha_{1m}}\cdots x_k^{\alpha_{km}}} \right).  \nonumber 
 \end{eqnarray}
 Applying the scale transformation, we obtain,
 \begin{equation}
     \Phi\left( \frac{\varsigma_1 y_1}{\left(\sigma_1 x_1\right)^{\alpha_{11}}\cdots \left(\sigma_k x_k\right)^{\alpha_{k1}}}, \cdots, \frac{\varsigma_m y_m}{\left(\sigma_1 x_1\right)^{\alpha_{1m}}\cdots \left(\sigma_k x_k\right)^{\alpha_{km}}} \right). \label{eq:e22a}
 \end{equation}
 As $\varsigma_1,\cdots,\varsigma_m$ are functionally dependent, using Eqs.~\ref{eq:e20} we have,
 \begin{eqnarray}
     \Phi\left( \frac{\varsigma_1 y_1}{\sigma_1^{\alpha_{11}} x_1^{\alpha_{11}}\cdots \sigma_k^{\alpha_{k1}} x_k^{\alpha_{k1}}}, \cdots, \frac{\varsigma_m y_m}{\sigma_1^{\alpha_{1m}} x_1^{\alpha_{1m}}\cdots \sigma_k^{\alpha_{km}} x_k^{\alpha_{km}}} \right)\nonumber \\
     =\Phi\left( \frac{\varsigma_1 y_1}{\varsigma_1 x_1^{\alpha_{11}}\cdots x_k^{\alpha_{k1}}}, \cdots, \frac{\varsigma_m y_m}{\varsigma_m x_1^{\alpha_{1m}}\cdots x_k^{\alpha_{km}}} \right) \nonumber \\= \Phi\left( \frac{y_1}{x_1^{\alpha_{11}}\cdots x_k^{\alpha_{k1}}}, \cdots, \frac{y_m}{x_1^{\alpha_{1m}}\cdots x_k^{\alpha_{km}}} \right) . \nonumber
 \end{eqnarray}
 Indeed, Eq.~\ref{eq:e23} does not change under scale transformations. 
\end{proof} 

Here I define the transformation of a physical function into a form invariant under all scale transformations of its parameters, namely the transformation from Eq.~\ref{eq:e18} to Eq.~\ref{eq:e23}, as the {\it exploitation of self-similarity}. The introduced new parameters $\Pi_1, \cdots, \Pi_m$ in Eq.~\ref{eq:e21} are called {\it similarity parameters} or {\it similarity variables} which are always power functions. Note that the similarity parameters are functions invariant under scale-transformation, 
\begin{eqnarray}
\Pi_i\left(\varsigma_i y_i,\sigma_1 x_1,\cdots,\sigma_k x_k\right) &=& \frac{\varsigma_i y_i}{\sigma_1^{\alpha_{1i}} x_1^{\alpha_{1i}}\cdots \sigma_k^{\alpha_{ji}} x_k^{\alpha_{ji}}} = \frac{\varsigma_i y_i}{\varsigma_i  x_1^{\alpha_{1i}}\cdots  x_k^{\alpha_{ji}}} = \frac{y_i}{x_1^{\alpha_{1i}}\cdots  x_k^{\alpha_{ji}}} \nonumber \\
&=& \Pi_i\left(y_i, x_1,\cdots, x_k\right), \label{eq:e21a} 
\end{eqnarray}
which are called {\it scale-invariant functions} and the exponents $\alpha_{1i}, \cdots, \alpha_{ji}$ are called {\it similarity exponents}. Thus, $\Phi\left(\Pi_1, \cdots, \Pi_m \right)$ in Eq.~\ref{eq:e23} is also scale-invariant function as $\Pi_1, \cdots, \Pi_m$ are scale-invariant functions. 

Scaling exponents of Eq.~\ref{eq:e19a} and similarity exponents of Eq.~\ref{eq:e21} are related through the functionally independence expressed in Eq.~\ref{eq:e20}. If we take logarithm,
\begin{eqnarray}
    \ln \varsigma_i &=& \alpha_{1i} \ln \sigma_1+\cdots +\alpha_{ki} \ln \sigma_k, \nonumber \\
    \gamma_{1i} \ln A_1+\cdots +\gamma_{ki} \ln A_k &=& \alpha_{1i}\left(\omega_{11}\ln A_1+\cdots +\omega_{k1} \ln A_k  \right)+\cdots +\nonumber \\ 
    & &\alpha_{ki}\left(\omega_{1k} \ln A_1+\cdots + \omega_{kk} \ln A_k \right) \nonumber 
\end{eqnarray}
where $i=1, \cdots, m$. Finally, we have the following,
\begin{equation}
\ln A_1\left(-\gamma_{1i} + \omega_{11}\alpha_{1i} +\cdots+\omega_{1k} \alpha_{ki} \right)+\cdots+\ln A_k \left(-\gamma_{ki} + \omega_{11}\alpha_{1i} +\cdots+\omega_{kk}\alpha_{ki}  \right) = 0. \label{eq:e21b}   
\end{equation}
The similarity exponents $\alpha_{ij}$ $(i=1,\ldots,k;\ j=1,\ldots,m)$ are determined by solving Eqs.~\ref{eq:e21b}, given the scaling exponents $\omega_{ij}$ $(i,j=1,\ldots,k)$ and $\gamma_{ij}$ $(i=1,\ldots,k;\ j=1,\ldots,m)$.

The resulting expression, in which the original parameters are replaced by similarity parameters, represents a form in which the self-similarity of the system has been exploited. In this representation, all data points collapse onto a lower-dimensional manifold. We refer to this phenomenon as {\it data collapse} or {\it master curve}. Figure ~\ref{fig:F1} shows the comparison between the representation of physical parameters and the representation of data collapse. All the deviations derived from the difference of scale have been normalized to a single line in Fig.~\ref{fig:F1} (b).

\section{Physical parameters and similarity of the first kind and the second kind\label{sec:SIV}}

In this section, we introduce the structure of physical parameters into the previous formulation. It will be shown that physical parameters consist of two components, namely numerical values and units. As a consequence, there are two distinct kinds of scale transformations. This distinction naturally leads to the concepts of similarity of the first kind and similarity of the second kind.

\subsection{Physical parameters}

As we have formulated the exploitation of self-similarity in terms of scale transformation. Here we introduce the physical parameters formulated in Ref.~\cite{Watanabe25}, which will lead to the interesting feature. The physical parameter $z$ in a system is represented as the product of a numerical value and its unit,
\begin{equation}
z_i = N\left[z_i, {\rm unit}_i \right]{\rm unit}_i,\label{eq:e25}
\end{equation}
where $N$ denotes the numerical value of $z_i$ measured in the unit $\mathrm{unit}_i$. This structure is easily understood by observing a physical parameter, e.g., $l=10~\mathrm{m}$, where $N = 10$ and $\mathrm{unit}={\rm m}$. Note that the numerical value is given by the function $N[Au, u] = A$ for any real number $A$. In the previous example, $N[10~\mathrm{m},\mathrm{m}]=10$. This function is essential for understanding why there are two kinds of self-similarity.

\subsection{Scale transformation of Units}

A scale transformation acting on units is called a scale transformation of units,
\begin{equation}
{\rm unit}_i \mapsto {\rm unit^{'}}_i = \frac{1}{\phi_i\left( L \right)} {\rm unit}_i \label{eq:e26}
\end{equation}
where $\phi_i\left(L \right)$ is called a dimension function, which determines the factors of the scale transformation of units $L$.

$\phi_i$ is a function that determines the factors associated with the scale transformation of units. For example, $ {\rm  m} \mapsto {\rm cm} = \frac{{\rm m}}{100}$ where $L = 100$, $ \frac{{\rm  km}}{ {\rm hour}} \mapsto\frac{{\rm m}  }{{\rm sec} }  = \frac{3600}{1000}\frac{{\rm km}}{{\rm hour}}$ where $L =\frac{1000}{3600}$.

Note that physical parameters themselves do not change under scale transformations of units, e.g., $1~\mathrm{m} = 100~\mathrm{cm}$. Therefore, the following identity must be satisfied:
\begin{equation}
    z_i = \phi_i N\left[z_i, \mathrm{unit_i} \right]\frac{{\rm unit_i}}{\phi_i} =   N\left[z_i, \mathrm{unit^\prime_i} \right]{\rm unit^\prime_i}
\end{equation}
where ${\rm unit^\prime_i} ={\rm unit_i}/\phi_i$. Therefore, the function $N$ has the following relation for a scale transformation of units, 
\begin{eqnarray}
    N\left[z_i, {\rm unit}_i \right] \mapsto N\left[z_i, {\rm unit^{'}}_i \right]  = N\left[z_i, \frac{{\rm unit}_i}{\phi_i} \right] \nonumber \\
    =\phi_i  N\left[z_i, {\rm unit}_i \right], \label{eq:e27}
\end{eqnarray}
while the identity is preserved for physical parameters\footnote{For example, the numerical values of the scale transformation of units of $3~{\rm m} \mapsto 300~{\rm cm}$ can be described as $N\left[ 3~{\rm m}, {\rm m} \right] \mapsto N\left[3~{\rm m},{\rm  cm} \right] = N\left[3~{\rm m},\frac{{\rm  m}}{100} \right] 
    = 100 N\left[ 3~{\rm m}, {\rm m} \right] = 300$, while, for the physical parameter,
    $N\left[ 3~{\rm m}, {\rm m} \right] {\rm m}   \mapsto N\left[3~{\rm m},{\rm  cm} \right] {\rm cm}  = N\left[3~{\rm m},\frac{{\rm  m}}{100} \right] \frac{{\rm  m}}{100} 
    = 100 N\left[ 3~{\rm m}, {\rm m} \right] \frac{{\rm  m}}{100} = N\left[ 3~{\rm m}, {\rm m} \right] {\rm m}$.},
\begin{eqnarray}
    z_i \mapsto z^{'}_i = N\left[z_i, {\rm unit^{'}}_i \right]{\rm unit^{'}}_i =N\left[z_i, \frac{{\rm unit}_i}{\phi_i} \right]\frac{{\rm unit}}{\phi_i} \nonumber \\
    = \phi_i N\left[z_i, {\rm unit}_i \right]\frac{{\rm unit}}{\phi_i} =N\left[z_i, {\rm unit}_i \right]{\rm unit} = z_i. \label{eq:e28} 
\end{eqnarray}


\subsection{Scale transformation of physical parameters}

A scale transformation acting on physical parameters is called a scale transformation of physical parameters,
\begin{equation}
    z_i \mapsto z^{'}_i = \sigma_i z_i. \label{eq:e29}
\end{equation}

A scale transformation of physical parameters changes the physical parameters. Thus, this can be described as
\begin{eqnarray}
    N[z_i, {\rm unit}_i] {\rm unit}_i \mapsto N[z^{'}_i, {\rm unit}_i] {\rm unit}_i = N[\sigma_i z_i, {\rm unit}_i] {\rm unit}_i = \sigma_i N[z_i, {\rm unit}_i] {\rm unit}_i. \label{eq:e31}  
\end{eqnarray}
This means that the numerical value $N$ changes under a scale transformation of physical parameters as well,
\begin{equation}
    N[z_i, {\rm unit}_i] \mapsto N[z^{'}_i, {\rm unit}_i] = N[\sigma_i z_i, {\rm unit}_i]  = \sigma_i N[z_i, {\rm unit}_i].  \label{eq:e30}
\end{equation}

For example, the transformation of the numerical value and the physical parameter under the scale transformation $ 3~{\rm m} \mapsto 300~{\rm m}$ can be desribed as
\begin{equation}
    N\left[3~{\rm m},{\rm m} \right] \mapsto N\left[300~{\rm m}, {\rm m} \right] = 300, \nonumber
\end{equation}
\begin{equation}
    N\left[3~{\rm m},{\rm m} \right]  {\rm m} \mapsto N\left[300~{\rm m}, {\rm m} \right] {\rm m} = 300~{\rm m}, \nonumber
\end{equation}
respectively. Note that the difference between a scale transformation of units and a scale transformation of physical parameters appears here. Both physical parameters and its numerical values change in this transformation while only numerical values change in the scale transformation of units. That is to say, numerical values do not distinguish between these two different scale transformations.

\subsection{Dimensional analysis}

Now we introduce dimensional analysis. As we have defined two scale transformations acting on different components of physical parameters, there are two ways of exploiting self-similarity. The important point is that {\it we know the scale functions of units a priori}, while we do not know the self-similar structures of physical parameters. This is because the scale functions of units are introduced through the measurement. Therefore, we can always transform any physical function into a form in which all the scale invariance associated with units is exploited. Another crucial point is that {\it numerical values do not distinguish between scale transformations of units and physical parameters}. For example, the change of unit $3~{\rm m} \mapsto 300~{\rm cm}$ and the scale transformation of the physical parameter $3~{\rm m} \mapsto 300~{\rm m}$ are fundamentally different operation. The former does not represent any physical change, whereas the latter corresponds to a genuine physical change. However, this distinction is not reflected in the corresponding numerical values $3 \mapsto 300$. Therefore, a transformed function that is invariant under a scale transformation of units is also invariant under scale transformations of physical parameters if the scale transformations of units and physical parameters share the same structure. Let us see this story in this subsection.

Suppose that $a_1, \cdots, a_k$ are numerical values of physical parameters with functionally independent scale functions, and $b_1,\cdots,b_m$ are numerical values of physical parameters with functionally dependent scale functions,
\begin{eqnarray}
    a_1 = N[{\rm a}_1, {\rm unit_a}_1],\cdots,a_k = N[{\rm a}_k, {\rm unit_a}_k], \nonumber \\ 
    b_1 = N[{\rm b}_1, {\rm unit_b}_1],\cdots,b_m = N[{\rm b}_m, {\rm unit_b}_m], \label{eq:e32}
\end{eqnarray}
where ${\rm a}_1, \cdots,{\rm a}_k$, ${\rm b}_1, \cdots,{\rm b}_m$, and ${\rm unit_a}_1,\cdots,{\rm unit_a}_k$, ${\rm unit_b}_1,\cdots,{\rm unit_b}_m$ are the corresponding physical parameters and their units, respectively. There are the following scale functions of units,
\begin{eqnarray}
    \phi_1 a_1 = N[{\rm a}_1, \frac{{\rm unit_a}_1}{\phi_1}],\cdots,\phi_k a_k = N[{\rm a}_k,\frac{{\rm unit_a}_k}{\phi_k}],\nonumber \\
    \varphi_1 b_1 = N[{\rm b}_1, \frac{{\rm unit_b}_1}{\varphi_1}],\cdots,\varphi_m b_m = N[{\rm b}_m, \frac{{\rm unit_b}_m}{\varphi_m}]. \label{eq:e33}
\end{eqnarray}
where $\phi_1, \cdots, \phi_k$ are functionally independent scale functions of units, and $\varphi_1,\cdots,\varphi_m$ are functionally dependent scale functions of units:
\begin{eqnarray}
\varphi_1 = \phi_1^{\alpha_{11}}\cdots\phi_k^{\alpha_{k1}}, 
\cdots, 
\varphi_m = \phi_1^{\alpha_{1m}}\cdots\phi_k^{\alpha_{km}}.
\label{eq:e33a}
\end{eqnarray}
There is the following function,
\begin{equation}
        F\left(a_1, \cdots, a_k, b_1,\cdots,b_m \right) = 0, \label{eq:e34}
\end{equation}
where scale transformation of units are defined as follows,
\begin{equation}
        F\left(\phi_1 a_1, \cdots,\phi_k a_k, \varphi_1 b_1,\cdots,\varphi_m b_m \right) = 0. \label{eq:e35}
\end{equation}
Now we define dimensional analysis as follows:
\begin{definition}
Dimensional analysis is the transformation of a function into a form invariant of scale transformation of {\it units} as follows,\begin{equation}
    \Phi \left(\Pi_1, \cdots, \Pi_m\right) = 0 \label{eq:e36}
\end{equation}
where
\begin{equation}
 \Pi_1 = \frac{b_1}{a_1^{\alpha_{11}} \cdots a_k^{\alpha_{k1}}}, \cdots, \Pi_m = \frac{b_m}{a_1^{\alpha_{1m}}\cdots a_k^{\alpha_{km}}}. \label{eq:e37}
\end{equation}
\end{definition}

When the scale functions are obtained, functions can be transformed into forms invariant under the scale transformation, as shown in Sec.\ref{subsec:D}.

Note that dimensional analysis is, in principle, the transformation of a function into a form invariant under scale transformations of units. However, the resulting form is also invariant under scale transformations of physical parameters in those parts where the scale invariance of units and that of physical parameters are equivalent.

\begin{proposition}\label{prop:p2}
If the scale invariance of physical parameters coincides with the scale invariance induced by units, in the sense that corresponding scale transformations produce identical scale factors for all physical quantities, then dimensional analysis exploits the scale invariance of the physical parameters.
\end{proposition}

\begin{proof}

Applying the scale transformation of units to the nondimensionalized function in Eq.~\ref{eq:e37}, we have
\begin{equation}
    \Phi\left(\frac{\varphi_1 b_1}{(\phi_1 a_1)^{\alpha_{11}} \cdots ( \phi_k a_k)^{\alpha_{k1}}}, \cdots, \frac{ \varphi_m b_m}{(\phi_1 a_1)^{\alpha_{1m}} \cdots ( \phi_k a_k)^{\alpha_{km}}} \right) \label{eq:e38}
\end{equation}
As $\varphi_1, \cdots, \varphi_m$ are functionally dependent scale functions, they can be replaced as Eq.~\ref{eq:e33a}. Therefore, we obtain
\begin{eqnarray}
    \Phi\left(\frac{\varphi_1 b_1}{(\phi_1 a_1)^{\alpha_{11}} \cdots (\phi_k a_k)^{\alpha_{k1}}}, \cdots, \frac{\varphi_m b_m}{(\phi_1 a_1)^{\alpha_{1m}} \cdots (\phi_k a_k)^{\alpha_{km}}} \right) \nonumber \\
    =\Phi\left(\frac{\varphi_1 b_1}{ \varphi_1 a_1^{\alpha_{11}} \cdots a_k^{\alpha_{k1}}}, \cdots, \frac{\varphi_m b_m}{ \varphi_m a_1^{\alpha_{1m}} \cdots  a_k^{\alpha_{km}}} \right) \nonumber \\
    =\Phi\left(\frac{ b_1}{ a_1^{\alpha_{11}} \cdots a_k^{\alpha_{k1}}}, \cdots, \frac{ b_m}{ a_1^{\alpha_{1m}}\cdots a_k^{\alpha_{km}}} \right). \label{eq:e40}
\end{eqnarray}
Hence, Eq.~\ref{eq:e36} is invariant under scale transformation of units.

On the other hand, here we introduce the scale transformation of physical parameters corresponding to Eqs.~\ref{eq:e32} as follows,
\begin{eqnarray}
    \sigma_1 a_1 = N[\sigma_1{\rm a}_1, {\rm unit_a}_1],\cdots,\sigma_k a_k = N[\sigma_k {\rm a}_k,{\rm unit_a}_k], \nonumber \\
    \varsigma_1 b_1 = N[\varsigma_1 {\rm b}_1, {\rm unit_b}_1],
    \cdots,\varsigma_m b_m = N[\varsigma_m {\rm b}_m, {\rm unit_b}_m]. \label{eq:e39a}
\end{eqnarray}

If the scale invariance of physical parameters coincides with the scale invariance induced by units, there exists a correspondence between the scale transformations of physical quantities and those induced by transformations of units. Therefore we have,
\begin{eqnarray}
  \sigma_1 = \phi_1, \cdots, \sigma_k = \phi_k, \varsigma_1 = \varphi_1,\cdots, \varsigma_m = \varphi_m. \label{eq:e40a}
\end{eqnarray}
Since Eq.~\ref{eq:e36} has been shown to be invariant under the scale transformations of units $\phi_1, \cdots, \phi_k,\varphi_1,\cdots,\varphi_m$, it is also invariant under the scale transformations of physical parameters $\sigma_1,\cdots,\sigma_k,\varsigma_1,\cdots,\varsigma_m$. Therefore, if the scale invariance of units coincides with the scale invariance of physical parameters, dimensional analysis exploits scale invariance of physical parameters as well.

\end{proof} 

The correspondence assumed in Proposition~\ref{prop:p2} also comes from the fact that the transformations of physical parameters and units are indistinguishable in terms of their numerical values. Specifically, for each pair of corresponding transformations, there are the following relations,
\begin{eqnarray}
    N\left[ {\rm a}_1, \frac{{\rm unit_a}_1}{\phi_1}  \right] = N\left[ \sigma_1 {\rm a}_1, {\rm unit_a}_1 \right],\cdots,N\left[ {\rm a}_k, \frac{{\rm unit_a}_k}{\phi_k}  \right] = N\left[ \sigma_k {\rm a}_k, {\rm unit_a}_k \right], \nonumber \\
    N\left[ {\rm b}_1, \frac{{\rm unit_b}_1}{\varphi_1}  \right] = N\left[ \varsigma_1 {\rm b}_1, {\rm unit_b}_1 \right],\cdots,N\left[ {\rm b}_m, \frac{{\rm unit_b}_m}{\varphi_m}  \right] = N\left[ \varsigma_m {\rm b}_m, {\rm unit_b}_m \right].
    \label{eq:e42a}
\end{eqnarray}
This is the essence of how and why dimensional analysis actually works on physical problems. Note that the scale functions associated with the two types of scale transformations do not need to be identical, but it is sufficient that a correspondence exists between them.

For instance, let us think about the drag force $F$ in a high velocity $U$ of a body in a gas where the diameter is $D$, the gas density is $\rho$ and the speed of sound is $c$ \cite{Barenblatt2003}. The drag force is expressed by the following function,
\begin{equation}
    F = f\left(\rho, U ,D,c \right) \label{eq:e42b}
\end{equation}
where their dimension function are $\phi_F=LMT^{-2}$, $\phi_\rho=L^{-3}M$, $\phi_U=LT^{-1}$, $\phi_D=L$, $\phi_c=LT^{-1}$ in $LMT$ unit. Based on these dimension functions, we have the following dimensionless form by dimensional analysis,
\begin{equation}
    \frac{F}{\rho U^2D^2}= f\left(\frac{U}{c} \right), \label{eq:e42c}
\end{equation}
which is invariant under scale transformation of units. On the other hand, considering the hydrodynamical equations, we can deduce the drag force as $F \sim \rho U^2 D^2$ \cite{Munson2021}. From this drag force equation, the following scale transformation of physical parameters are naturally deduced, $\sigma_F = \mathcal{F}$, $\sigma_\rho=\mathcal{F} \mathcal{V}^{-2} \mathcal{L}^{-2}$, $\sigma_U=\mathcal{V}$, $\sigma_D=\mathcal{L}$ and $\sigma_c=\mathcal{V}$. These transformations conrrespond to scaling the drag force, the speed, and the size of sphere, in the physical contexts. Scale functions of each transformations are seemingly different. However, you can find the following correspondence between two transformations, $LMT^{-2}=\mathcal{F}$, $LT^{-1} = \mathcal{V}$ and $L=\mathcal{L}$. Therefore, Eq.~\ref{eq:e42c} is invariant under scale transformation of physical parameters as well.

\subsection{Similarity of the first kind and the second kind}

Dimensional analysis is fundamentally a transformation into a form invariant under scale transformations of units, but it is also applicable to scale transformations of physical parameters due to the indistinguishability of numerical values.

Proposition \ref{prop:p2} implies that dimensional analysis is fundamentally a transformation into a form invariant under scale transformations of units, but it is also applicable to scale transformations of physical parameters due to the indistinguishability of numerical values. Thus, dimensional analysis exploits all scale invariance if and only if the scale transformations induced by units and those of physical parameters coincide. Therefore, we naturally arrive at the distinction of whether the scale invariance of physical parameters and that of units are equivalent or not. Here we define similarity of the first kind and similarity of the second kind as follows.
\begin{definition}\label{dfn:d3}
Suppose that $\boldsymbol{\phi}$, $\boldsymbol{\sigma}$, and $\boldsymbol{\xi}$ are sets of scale functions. Two sets of scale functions are regarded as equivalent if there exists a correspondence between their scale transformations such that the corresponding transformations yield identical scale factors for all physical quantities, and vice versa\footnote{The mathematical condition for the existence of such a correspondence is discussed in Ref.~\cite{Watanabe25} and Appendix~\ref{sec:appendix_Ca}.}. If the scale transformation of units $\boldsymbol{\phi}$ and the scale transformation of physical parameters $\boldsymbol{\sigma}$ are equivalent, namely there are the correspondence between two scale transformations, the problem is said to exhibit similarity of the first kind, or complete similarity,
\begin{equation}
    \boldsymbol{\phi} \equiv \boldsymbol{\sigma}. \label{eq:e41}
\end{equation}
On the other hand, if the scale transformation of units $\phi$ and the scale transformation of parameters $\boldsymbol{\sigma}$ are not equivalent, but there exists an intrinsic scale transformation $\boldsymbol{\xi}$ derived from $\boldsymbol{\sigma}$, which cannot be generated by the scale transformation of units alone, the problem is said to exhibit similarity of the second kind, or incomplete similarity:
\begin{equation}
    \boldsymbol{\sigma} \equiv \boldsymbol{\xi \phi} \label{eq:e42}
\end{equation}
where $\boldsymbol{\xi \phi}$ means the scale functions obtained by combining the unit-induced scaling $\boldsymbol{\phi}$ with the intrinsic scaling $\boldsymbol{\xi}$.
\end{definition}

This definition of similarity of the first kind and similarity of the second kind is partially equivalent to that of Barenblatt \cite{Barenblatt1996,Barenblatt2003}, as discussed in Appendix \ref{sec:appendix_C}.

\subsection{Self-similarity in general}

Data collapse of general case can be formulated as follows. Consider the following functions,
\begin{equation}
    F\left(a_1,\cdots,a_k, b_1, \cdots, b_l,c_1,\cdots,c_p \right) = 0, \label{eq:e54}
\end{equation}
and there are three types of numerical values of physical parameters,
\begin{itemize}
    \item {\it Numerical values of physical parameters with functionally independent scale functions of units}
    \begin{equation}
        \phi_1 a_1 = N\left[{\rm a}_1, \frac{{\rm unit_a}_1}{\phi_1 }\right],\cdots,\phi_k a_k = N\left[{\rm a}_k,\frac{ {\rm unit_a}_k}{\phi_k}\right], \label{eq:e55}
    \end{equation}
    \begin{equation}
        \phi_1 = L_1^{\omega_{11}}\cdots L_k^{\omega_{k1}},\cdots,\phi_k = L_1^{\omega_{1k}}\cdots L_k^{\omega_{kk}}. \label{eq:e55a}
    \end{equation}
    \item {\it Numerical values of physical parameters with functionally dependent scale functions of units and functionally independent scale functions of physical parameters}
    \begin{equation}
     \varphi_1 \varsigma_1 b_1 = N\left[ \varsigma_1 {\rm b}_1, \frac{{\rm unit_b}_1}{\varphi_1}\right],
    \cdots,\varphi_l \varsigma_l b_l = N\left[\varsigma_l{\rm b}_l, \frac{{\rm unit_b}_l}{\varphi_l}\right], \label{eq:e56}
    \end{equation}
    \begin{equation}
      \varphi_1 = L_1^{\gamma_{11}}\cdots L_k^{\gamma_{k1}},\cdots, \varphi_l = L_1^{\gamma_{1l}}\cdots L_k^{\gamma_{kl}}, 
      \varsigma_1 = A_1^{\chi_{11}}\cdots A_l^{\chi_{l1}},\cdots, \varsigma_l = A_1^{\chi_{1l}}\cdots A_l^{\chi_{ll}}. \label{eq:e56a}
    \end{equation}
    \item {\it Numerical values of physical parameters with functionally dependent scale functions of units and physical parameters}
    \begin{eqnarray}
        \varrho_1 \xi_1 c_1 = N\left[\xi_1 {\rm c}_1, \frac{{\rm unit_c}_1}{\varrho_1}\right],
       \cdots, \varrho_p \xi_p c_p = N\left[\xi_p {\rm c}_p, \frac{{\rm unit_c}_p}{\varrho_p}\right], \label{eq:e57}
    \end{eqnarray}
    \begin{equation}
      \varrho_1 = L_1^{\iota_{11}}\cdots L_k^{\iota_{k1}},\cdots, \varrho_p = L_1^{\iota_{1p}}\cdots L_k^{\iota_{kp}},
      \xi_1 = A_1^{\lambda_{11}} \cdots A_l^{\lambda_{l1}},\cdots, \xi_p = A_1^{\lambda_{1p}}\cdots A_l^{\lambda_{lp}}. \label{eq:e57a}
    \end{equation}
\end{itemize}
Note that the scale functions $\phi_1, \cdots, \phi_k,\varphi_1, \cdots, \varphi_p$ can also be derived from scale transformations of physical parameters, such as $\phi_1 a_1 = N\left[ \phi_1 {\rm a}_1, {\rm unit_a}_1 \right]$. In numerical values, the origins of the scale functions are indistinguishable, whereas these origins are {\it physically} quite different. However, the scale functions $\phi_1, \cdots, \phi_k,\varphi_1, \cdots, \varphi_p$ can be obtained {\it a priori} from the units. Due to the indistinguishability of the origin of scale transformations in numerical values, we may assume that the scale functions deduced from the structure of units are also shared by the physical parameters. In order to distinguish the different operation, we denote $\sigma_\phi$ as the scale transformation of physical parameters sharing the structure with the scale transformation of units $\phi$.

$\phi_1,\cdots, \phi_k$, $\varsigma_1,\cdots, \varsigma_l$ are functionally independent, therefore there are following relation,
\begin{eqnarray}
  \varphi_1 = \phi_1^{\alpha_{11}}\cdots \phi_k^{\alpha_{k1}},\cdots, \varphi_l = \phi_1^{\alpha_{1l}}\cdots \phi_k^{\alpha_{kl}}, \label{eq:e57b} \\
  \varrho_1 = \phi_1^{\beta_{11} }\cdots \phi_k^{\beta_{k1}},\cdots,\varrho_p = \phi_1^{\beta_{1p} }\cdots \phi_k^{\beta_{kp}}, \label{eq:e57c} \\
  \xi_1 = \varsigma_1^{\kappa_{11}}\cdots\varsigma_l^{\kappa_{l1}},\cdots,\xi_p=\varsigma_1^{\kappa_{1p}}\cdots\varsigma_l^{\kappa_{lp}}. \label{eq:e57d}
\end{eqnarray}

Thus we have the following self-similar structure:
\begin{eqnarray}
    F\left(\phi_1 a_1,\cdots,\phi_k a_k, \varphi_1 \varsigma_1 b_1, \cdots,\varphi_l \varsigma_l b_l, \varrho_1 \xi_1 c_1,
    \cdots,\varrho_p \xi_p c_p \right) = 0.\label{eq:e58}
\end{eqnarray}
Introducing the following similarity variables,
\begin{eqnarray}
   \Pi_1 = \frac{b_1}{a_1^{\alpha_{11}}\cdots a_k^{\alpha_{k1}}},\cdots, \Pi_l = \frac{b_l}{a_1^{\alpha_{1l}}\cdots a_k^{\alpha_{jl}}},~\nonumber \\
   \Xi_1 = \frac{c_1}{a_1^{\beta_{11}}\cdots a_k^{\beta_{k1}}},\cdots, 
   \Xi_p=\frac{c_p}{a_1^{\beta_{1p}}\cdots a_k^{\beta_{kp}}}. \label{eq:e59}
\end{eqnarray}
Then we have,
\begin{eqnarray}
    F\left(a_1,\cdots,a_k, a_1^{\alpha_{11}}\cdots a_k^{\alpha_{k1}} \Pi_1, \cdots, a_1^{\alpha_{1p}}\cdots a_k^{\alpha_{kp}} \Xi_p \right) \nonumber \\
    = F\left(a_1,\cdots,a_k, \Pi_1, \cdots, \Xi_p \right). \label{eq:e60}
\end{eqnarray}
As $a_1, \cdots,a_k$ are parameters with functionally independent scale functions, these parameters can be considered as independent of function $F$ according to Theorem \ref{thm:2-3}. Therefore, we obtain
\begin{equation}
    \Phi\left(\Pi_1,\cdots,\Pi_l, \Xi_1, \cdots, \Xi_p \right) = 0. \label{eq:e61}
\end{equation}
However, Eq.~\ref{eq:e61} possesses the scale transformation,
\begin{equation}
    \Phi\left(\varsigma_1 \Pi_1,\cdots,\varsigma_l \Pi_l, \xi_1 \Xi_1, \cdots, \xi_p \Xi_p \right) = 0. \label{eq:e62}
\end{equation}
Let us introduce the following similarity variables 
\begin{equation}
    Z_1= \frac{\Xi_1}{\Pi_1^{\kappa_{11}}\cdots\Pi_l^{\kappa_{l1}} },\cdots, Z_p= \frac{\Xi_p}{\Pi_1^{\kappa_{1p}}\cdots\Pi_l^{\kappa_{lp}} }\label{eq:e63}
\end{equation}
where $\kappa_{11},\cdots,\kappa_{lp}$ are determined by solving following,
\begin{equation}
    -\lambda_{ik} + \chi_{ij}\kappa_{jk} = 0,~\left(i,j=1,\cdots,l,~k= 1,\cdots,p  \right)
\end{equation}
using given scaling exponents $\lambda_{ik}, \chi_{ij}$, which are derived from the functionally dependence in Eq.~\ref{eq:e57b}, Eq.~\ref{eq:e57c} and Eq.~\ref{eq:e57d}.

Introducing Eqs.~\ref{eq:e63} into Eq.~\ref{eq:e61}, then we obtain
\begin{eqnarray}
    \Phi\left(\Pi_1,\cdots,\Pi_l,  \Pi_1^{\kappa_{11}}\cdots\Pi_l^{\kappa_{l1}} Z_1, \cdots,\Pi_1^{\kappa_{1p}}\cdots\Pi_l^{\kappa_{lp}} Z_p \right) \nonumber \\
    = \Phi\left(\Pi_1,\cdots,\Pi_l, Z_1, \cdots, Z_p \right). \label{eq:e65}
\end{eqnarray}
As there is a scale transformation under which a variables among $\Pi_1,\cdots,\Pi_l$ changes while the other variables remain unchanged, according to Theorem \ref{thm:2-3}, therefore we obtain,
\begin{equation}
\Phi\left(Z_1, \cdots, Z_p \right) = 0. \label{eq:e66}       
\end{equation}
Defining $\boldsymbol{\phi}=({\varrho_n},{\varphi_j},{\phi_i})$, $\boldsymbol{a}= ({c_n},{b_j},{a_i})$, $\boldsymbol{\Pi} = ({\Xi_n},{\Pi_j})$, $\boldsymbol{\xi}=({\xi_n},{\varsigma_j})$,$\boldsymbol{Z}=({Z_n})$ where $(n=1,\cdots,p;j=1,\cdots,l; i=1,\cdots, k)$, we can see the following hierarchical structure:
\begin{equation}
  \Phi\left[ \boldsymbol{Z}\left(\boldsymbol{\xi} \boldsymbol{\Pi}\left( \boldsymbol{\phi} \boldsymbol{a} \right)  \right)\right] =\Phi\left[ \boldsymbol{Z}\left(\boldsymbol{\Pi}\left(  \boldsymbol{a} \right)  \right)\right]. \label{eq:e67}  
\end{equation}
where multiplication between bold symbols denotes component-wise multiplication,
\begin{eqnarray}
\boldsymbol{\phi} \boldsymbol{a} &=& (\varrho_1 c_1,\cdots,\varrho_p c_p,\varphi_1 b_1,\cdots,\varphi_l b_l,\phi_1 a_1,\cdots, \phi_k a_k  ) \nonumber \\
\boldsymbol{\xi} \boldsymbol{\Pi} &=& (\xi_1 \Xi_1,\cdots, \xi_p \Xi_p, \varsigma_1 \Pi_1, \cdots, \varsigma_l \Pi_l ). \nonumber
\end{eqnarray}

Here we call the similarity parameters $\Pi$ {\it similarity parameters  of the first class}, similarity parameters $Z$ {\it similarity parameters of the second class}, respectively, to reflect the hierarchical structure (Fig.~\ref{fig:F2a}) \cite{Maruoka23,Watanabe25}. The hierarchical structure of invariance of each scale transformations are summarized in the Table \ref{tb:t1}. As you can see that the class of similarity variables reflect its scale-invariance of each scale functions. Similarity variables of the first class are invariant under scale transformation of units $\boldsymbol{\phi}$ and physical parameters whose invariance is equivalent to units $\boldsymbol{\sigma_\phi}=({\varrho_n},{\varphi_j},{\phi_i})_\sigma$\footnote{ Remind that $\boldsymbol{\sigma_\phi}$ is the scale transformation of physical parameters, not units while their scale functions are the same with $\boldsymbol{\phi}$.} while it is not invariant under the intrinsic scale transformation of physical parameters $\boldsymbol{\xi}$. Similarity variables of the second class are invariant under all kinds of scale transformation $\boldsymbol{\phi}$, $\boldsymbol{\sigma}$. Self-similar solutions of the first kind are lacking in similarity variables of the second class.
\begin{figure}
\centering
  \includegraphics[width=8.5 cm]{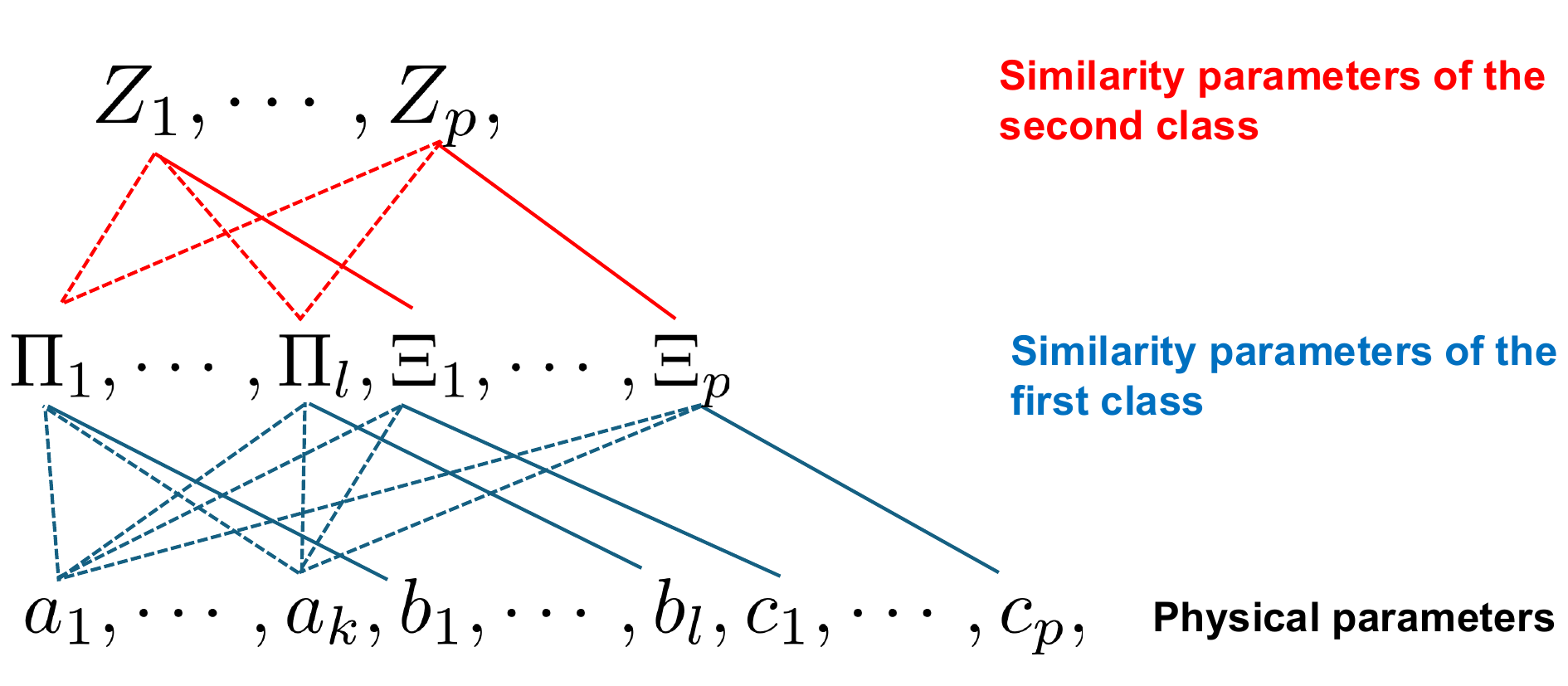} %
 \caption{Hierarchical structure of similarity parameters. Physical parameters $a_1, \cdots, a_k, b_1, \cdots, b_l, c_1, \cdots, c_p$ are quantities composed of a numerical value and a physical unit. They are not invariant under scale transformations of units. $\Pi_1, \cdots, \Pi_l, \Xi_1, \cdots, \Xi_p$ are referred to as similarity parameters of the first class, which are composed of physical parameters and are invariant under scale transformations of {\it units}. $Z_1, \cdots, Z_p$ are referred to as similarity parameters of the second class, which are composed of similarity parameters of the first class and are invariant under all scale transformations. The dashed lines indicate parameters with functionally independent scale functions, whereas the solid lines indicate the remaining parameters with functionally dependent scale functions. The terms {\it first class} and {\it second class} refer to the hierarchical levels of similarity parameters and should not be confused with similarity of the first kind and similarity of the second kind \cite{Maruoka23}.} 
  \label{fig:F2a}
\end{figure}

The forms of Eq.~\ref{eq:e67} and Eq.~\ref{eq:e66}, where their variables are replace to similarity variables of the second class, are invariant under any scale transformation, including scale transformations of {\it units} and of {\it physical parameters}. We say that {\it the scale invariance of the system has been fully exploited in $\Phi$}. Let us call $\Phi$ {\it fully exploited scaling function}.

\begin{table}[h]
\centering
\begin{tabular}{l|llll}
\hline
 & $\boldsymbol{\phi}$ & $ \boldsymbol{\sigma} $& $\boldsymbol{\sigma_\phi}$ & $\boldsymbol{\xi}$ \vspace{.5mm}\\  
\hline
 ${\rm unit}$ & $\bigcirc$ & $\times$ & $\times$ & $\times$ 
 \vspace{.5mm}\\
Numerical values & $\bigcirc$ & $\bigcirc$  & $\bigcirc$ & $\bigcirc$  \vspace{.5mm}\\
physical parameters & $\times$ & $\bigcirc$ & $\bigcirc$ & $\bigcirc$ \vspace{.5mm}\\
similarity parameters of the first class &  $\times$ &  & $\times$ & $\bigcirc$  \vspace{.5mm} \\
similarity parameters of the second class &$\times$ & $\times$ & $\times$ & $\times$ \vspace{.5mm}\\  
\hline
\end{tabular}
\caption{The table indicates whether the quantities change ($\bigcirc$) or remain unchanged ($\times$) under the application of scale transformations. Here, $\boldsymbol{\phi}$ denotes the scale transformation of units, $\boldsymbol{\sigma}$ denotes the scale transformation of physical parameters, $\boldsymbol{\sigma_\phi}$ represents the scale transformation of physical parameters corresponding to that of the units, and $\boldsymbol{\xi}$ denotes the intrinsic scale transformation of physical parameters.}
\label{tb:t1}
\end{table}

Note that each parameters $Z_1,\cdots,Z_p$ are functionally dependent. Therefore, there exists an explicit function of the form,
\begin{equation}
    Z_1 = \Phi_\mathrm{ex}\left(Z_2,\cdots,Z_p \right). \label{eq:e72} 
\end{equation}
while $\Phi_\mathrm{ex}$ cannot be a scale function. In other words, $\Phi_\mathrm{ex}$ cannot be expressed as a power function. That is, there exists no relation of the form,
\begin{equation}
    \Phi_\mathrm{ex}\left(Z_2,\cdots, Z_p \right) \neq Z_2^{c_2} \cdots Z_p^{c_p}. \label{eq:e73}
\end{equation}
This is because all the scale invariance has already been exploited in $\Phi_\mathrm{ex}$. As the power function is scale invariant, if $\Phi_\mathrm{ex}$ could be expressed in the form of power function Eq.~\ref{eq:e73}, some scale invariance would still remain unexploited.

From this formulation, we can now easily understand that Case I in Sec.~\ref{sec:SII} belongs to similarity of the first kind, as its self-similar form is obtained by dimensional analysis. On the other hand, Case II and Case III belong to similarity of the second kind, since their self-similar forms cannot be obtained by dimensional analysis and involve intrinsic scale transformations.

For example, the self-similar structure of case II is as follows \cite{Maruoka23}, 
\begin{equation}
    F \left(\delta_m, R, h, \rho, \mu, E, v_i \right) = 0, \label{eq:e43} 
\end{equation}
where the dimension functions for each variables are,
\begin{eqnarray}
    \phi_{\delta_m} = L,~ \phi_R=L, \phi_h= L, ~\phi_\rho=L^{-3}M, \nonumber \\
    \phi_\mu= L^{-1}MT^{-1}, ~\phi_E= L^{-1}MT^{-2}, ~\phi_{v_i}=LT^{-1}. \label{eq:e44} 
\end{eqnarray}
Based on Eq.~\ref{eq:e44}, dimensional analysis yields the following dimensionless variables\footnote{Similarity exponents are deduced by choosing $R$, $\rho$ and $E$ as parameters with functionally independent scale functions.}:
\begin{equation}
    \Pi = \frac{\delta_m}{R},~\kappa=\frac{h}{R},~\theta=\frac{\mu}{E^{1/2}\rho^{1/2}R},~\eta= \frac{\rho v_i^2}{E}. \label{eq:e45} 
\end{equation}
Thus, Eq.~\ref{eq:e43} is transformed into the form invariant of scale transformation of units,
\begin{equation}
    \Phi\left(\Pi, \kappa, \theta, \eta \right)=0. \label{eq:e46}
\end{equation}
However, Eq.~\ref{eq:e46} still possesses the 
following scale transformation\footnote{This scale transformation of Eq.~\ref{eq:e47} and Eq.~\ref{eq:e49b} which appears later are deduced by solving the model, Maxwell viscoelastic foundation model \cite{Maruoka23}. The important point is that scale transformation inherent to physical parameters cannot be deduced by dimensional analysis but other paths, such direct solutions, the renormalization group method\cite{Goldenfeld1989,Goldenfeld1992,CGO1996, Okumura2025,OKUMURA2026}, nonlinear eigenvalue problem \cite{Barenblatt1996, Barenblatt2003, Barenblatt2014, Eggers2015} etc.}:
\begin{equation}
    \xi_\Pi=A,~\xi_\kappa=A^3B^{-1},~\xi_\theta =AB^{-1/2} , ~\xi_\eta = B. \label{eq:e47}
\end{equation}
Based on the scale function of $\boldsymbol{\xi}$, we have the following similarity variables\footnote{Here $\Pi$ and $\eta$ are similarity parameters with functionally independent scale functions.},
\begin{equation}
    \Psi = \frac{\Pi^3}{\kappa \eta},~Z = \frac{\Pi}{\theta \eta^{1/2}}. \label{eq:e48}
\end{equation}
Finally, we obtain
\begin{equation}
    \Phi\left( \Psi, Z \right) =0. \label{eq:e49}
\end{equation}
Since $\Psi$ and $Z$ are functionally dependent, there exists a function $\Psi = \Phi_\mathrm{ex}\left(Z \right)$,
\begin{equation}
    \Phi_\mathrm{ex} \left(Z \right) = \frac{2}{3} \frac{Z}{1- \exp\left(-Z\right) }. \label{eq:e49b}
\end{equation}
In the end, we find that there is a following self-similar structure:
\begin{eqnarray}
\boldsymbol{a} &=& \left\{\delta_m, R, h, \rho, \mu, E, v_i  \right\}, \label{eq:50a} \\
\boldsymbol{\sigma} &=& \left\{ AL, L, A^3B^{-1}L, L^{-3}M, AB^{-1/2}L^{-1}MT^{-1}, \right. \nonumber \\
&& \left. L^{-1}MT^{-2}, BLT^{-1} \right\} \label{eq:e50} \\
\boldsymbol{\phi} &=& \left\{ L, L, L, L^{-3}M, L^{-1}MT^{-1}, L^{-1}MT^{-2}, LT^{-1} \right\} \label{eq:e51} \\
\boldsymbol{\xi} &=& \left\{ A, 1, A^3B^{-1}, 1, AB^{-1/2}, 1, B \right\} \label{eq:e52}
\end{eqnarray}
\begin{eqnarray}
    F(\boldsymbol{\sigma} \boldsymbol{a} )= F(\boldsymbol{\xi} \boldsymbol{\phi} \boldsymbol{a} ) =  F \left(AL\delta_m, LR,A^3B^{-1}L h,L^{-3}M \rho,  \right. \nonumber \\ 
    \left. AB^{-1/2} L^{-1}M T^{-1} \mu,  L^{-1}M T^{-2} E,BLT^{-1} v_i \right) = 0. \label{eq:e53}   
\end{eqnarray}
Figure \ref{fig:F2} shows the comparison among the physical parameters in Eq.~\ref{eq:e43}, similarity parameters of the first class in Eq.~\ref{eq:e45}, and the similarity parameters of the second class in Eq.~\ref{eq:e48}. The scale invariance is only partially exploited in the representation shown in Fig.~\ref{fig:F2} (b), whereas all the data points collapse onto a single curve in the representation of Eq.~\ref{eq:e48} shown in Fig.~\ref{fig:F2} (c). The solid line in Fig.~\ref{fig:F2} (c) corresponds to Eq.~\ref{eq:e49b}. Note that the function $\Phi_\mathrm{ex}$, in which all scale invariance has been exploited, is not a power-law function, as discussed in Eq.~\ref{eq:e73}. 
\begin{figure*}
\centering
\includegraphics[width=12 cm]{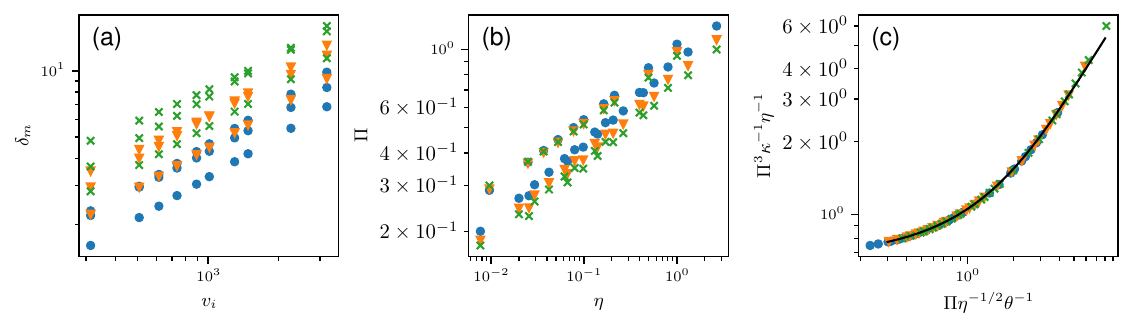} %
\caption{Plots of the physical parameters $\delta_\mathrm{m}$ versus $v_i$ (a), the similarity parameters of the first class, $\Pi=\delta_\mathrm{m}/R$ versus $\eta=\rho v_i^2/E$ (b), and the similarity parameters of the second class, $\Pi^3\kappa^{-1}\eta^{-1}$ versus $\Pi\eta^{-1/2}\theta^{-1}$ (c), for various physical conditions of the dynamical impact of a solid sphere onto a viscoelastic board \cite{Maruoka23}, where $\kappa=h/R$ and $\theta=\mu/(E^{1/2}\rho^{1/2}R)$. The solid line in (c) is described by Eq.~\ref{eq:e49b}.
}
\label{fig:F2}
\end{figure*}

The hierarchical structures in diffusion solution in Eq.~\ref{eq:e3a} and a dynamical impact of a solid sphere onto a viscoelastic board in Eq.~\ref{eq:e6a} are visualized in Fig.~\ref{fig:F7}. Similarity parameters consists of variables of the lower class, similarity parameters of the second class constitutes both physical parameters and similarity parameters of the first class. Note that parameters with functionally independent scale functions are shared among the similarity variables. However, this choice is arbitrary as far as they are functionally independent. Therefore, there are numerous way of nondimensionalization, which is discussed in Appendix \ref{sec:appendix_Ca}.

\begin{figure}[t]
    \centering
    \begin{minipage}[b]{0.48\linewidth}
        \centering
        \includegraphics[width=\linewidth]{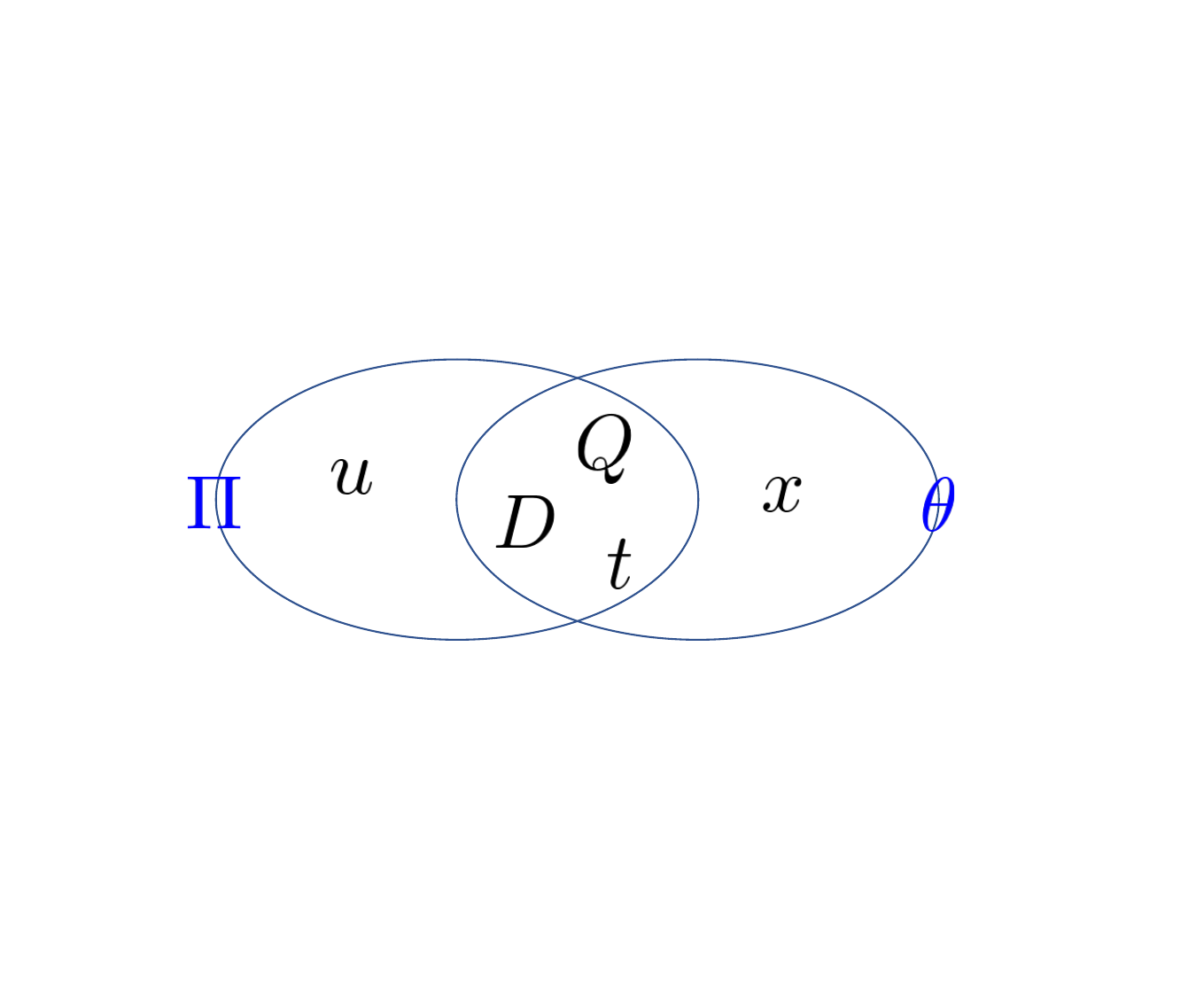}
        \subcaption{Hierarchical structure of self-similar solution of the first kind for diffusion in Eq.~\ref{eq:e3a}.}
        \label{fig:7a}
    \end{minipage}
    \hfill
    \begin{minipage}[b]{0.48\linewidth}
        \centering
        \includegraphics[width=\linewidth]{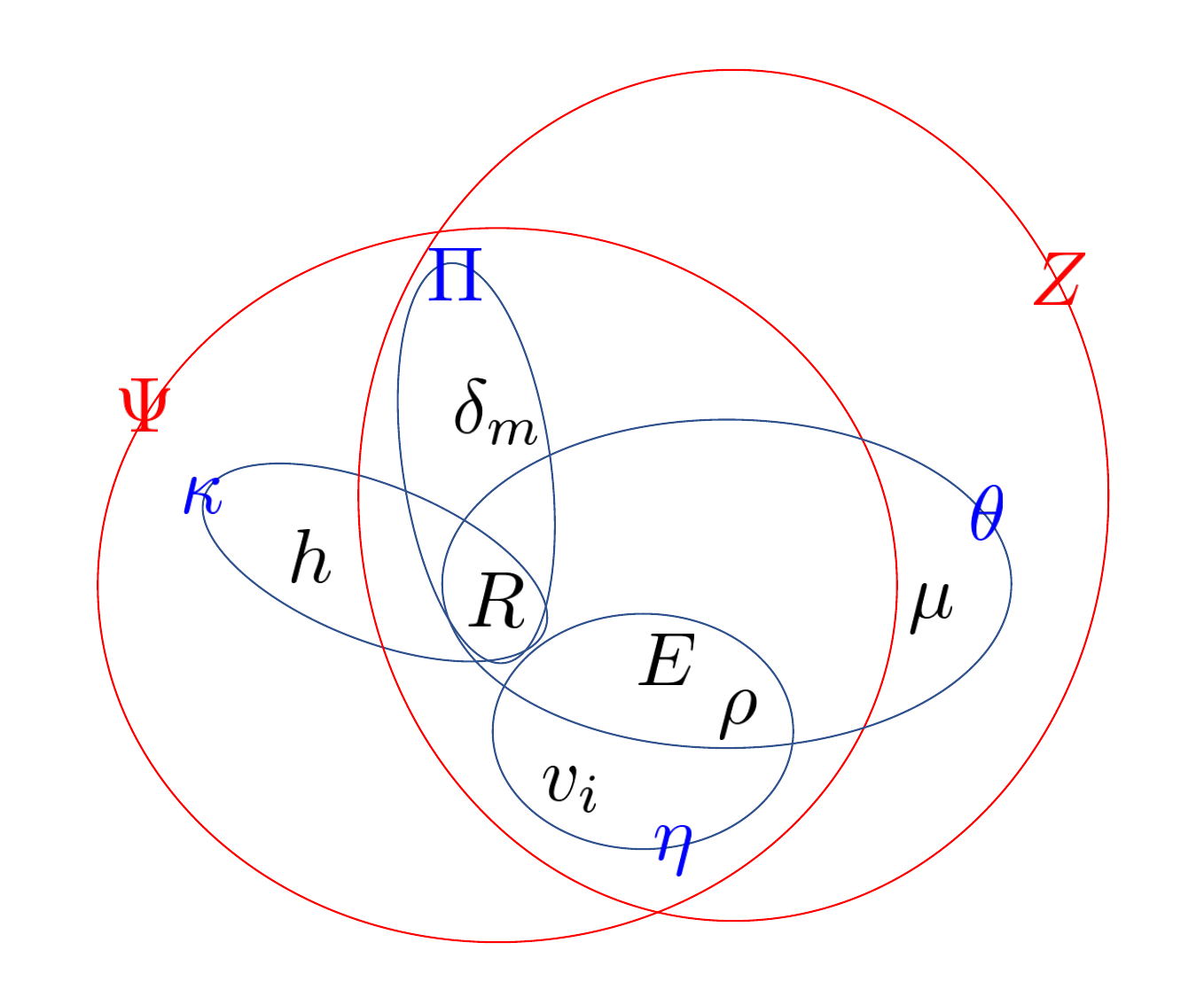}
        \subcaption{Hierarchical structure of self-similar solution of the second kind for a dynamical impact of a solid sphere onto a viscoelastic board in Eq.~\ref{eq:e6a}.}
        \label{fig:7b}
    \end{minipage}

    \caption{Hiearchical structre of self-similar solution of the first kind and the second kind. The parameters in the blue circles are the similarity paramteres of the first class, the parameters in the red circles constitutes the similarity parameters of the second class.}
    \label{fig:F7}
\end{figure}

\section{Self-similar solution of the second kind Type A and Type B \label{sec:SV}}

In the previous section, it was demonstrated that self-similarity can be classified according to how the scale invariance of physical parameters and that of units overlap. Self-similar solutions of the second kind are solutions that involve intrinsic scale invariance, which cannnot be fully exploited by only dimensional analysis. However, there are cases in which scaling and similarity exponents are fixed constants and cases in which they can be expressed as functions of other dimensionless parameters, can ranges continuous values:
\begin{equation}
     \varsigma_i\left( A_1,\cdots, A_k \right) =A_1^{\chi\left( c\right)_{i1}}\cdots A_k^{\chi\left( c\right)_{ki}} \label{eq:e74b}
\end{equation}
where $c$ is dimensionless number. Therefore, the similarity parameters are constructed as,
\begin{equation}
    Z_1= \frac{\Xi_1}{\Pi_1^{\kappa \left( c \right)_{11}}\cdots\Pi_l^{\kappa\left( c \right)_{l1}} },\cdots, Z_p= \frac{\Xi_p}{\Pi_1^{\kappa\left( c \right)_{1p}}\cdots\Pi_l^{\kappa\left( c \right)_{lp}} },\label{eq:e74}
\end{equation}
while both exponents are related as $   -\lambda(c)_{ik} + \chi(c)_{ij}\kappa(c)_{jk} = 0,~\left(i,j=1,\cdots,l,~k= 1,\cdots,p  \right)$. Therefore, scale functions are also parameter-dependent, and here we call these parameter-dependent scale transformation $\boldsymbol{\xi(c)} = (\xi(c)_n, \varsigma(c)_j)$ {\it parameter-dependent intrinsic scale transformation of physical parameters}. This case corresponds to the Barenblatt solution of the modified ground water problem \cite{Barenblatt2003}, the time evolution of converging shock wave \cite{Guderley1942, Landau1987}. Here we define a further subdivision of self-similar solution of the second kind as follows,
\begin{definition}\label{dfn:d4}
Suppose that a problem belongs to similarity of the second kind. If the scaling exponents and similarity exponents are fixed constant, independent of other parameters, namely if it involves {\it parameter-independent intrinsic scale transformation of physical parameters} $\boldsymbol{\xi^\prime}$, the corresponding solution is called a self-similar solution of the second kind Type A. On the other hand, if the exponents are not fixed constant but are dependent on additional dimensionless parameters, namely if it involves {\it parameter-dependent intrinsic scale transformation of physical parameters} $\boldsymbol{\xi(c)}$, the corresponding solution is called a self-similar solution of the second kind Type B. 
\end{definition}

The origin of this division can be understood by the formulation by Eggers and Fontelos where they distinguish the kinds of self-similarity according to whether self-similar solutions exist for a particular pair of exponents or for a continuous range of exponents whose values are determined by a regularity condition\footnote{See p.46 in Ref.~\cite{Eggers2015}}. 


For instance, let us introduce the ansatz $t^\alpha f\left(\frac{x}{t^\beta} \right)$ into the PDE involves nonlinear term $\partial_t u = u^2$, we obtain, 
\begin{equation}
  \alpha t^{\alpha-1}f\left(\zeta \right)- \beta t^{\alpha-1}\zeta \frac{ \partial f\left(\zeta \right)}{\partial \zeta} = t^{2 \alpha}  f^2 \left(\zeta \right).  \label{eq:e76a}    
\end{equation}
where $\zeta = x/t^\beta$. Eq. \ref{eq:e76a} should be the differential equation of $\zeta$, therefore $t$ must vanish, thus $\alpha=-1$ then we obtain,
\begin{equation}
  - f\left(\zeta \right) - \beta \zeta \frac{ d f\left(\zeta \right)}{d \xi} =  f^2\left(\zeta \right). \label{eq:e76}   
\end{equation}
Note that the power exponents $\alpha =-1$ is uniquely determined, whereas the other power exponent $\beta$ cannot be determined solely by introducing the ansatz. Instead, $\beta$ must be matched with another dimensionless number by a regularity condition \cite{Eggers2015}. The probelms of Type B typically involve differential equations including the similarity exponents. Such power exponents may be determined from a nonlinear eigenvalue problem \cite{Barenblatt1996,Barenblatt2003} or the renormalization group method \cite{Goldenfeld1992,CGO1996,Kunihiro2022,Okumura2025,OKUMURA2026}. In this case, the power exponents are undetermined by the local structure of equation. Such a case where the similarity exponents are matched with regularity conditions, are characteristic for self-similar solutions of the second kind Type B \cite{Barenblatt2003,Guderley1942,Weizsacker1954}. However, there are self-similar solutions whose similarity exponents are determined by the regularity conditions but their exponents are fixed \cite{Jensen1995}. Determining the similarity exponents by matching regularity conditions does not necessarily lead to parameter-dependent exponents, but can also yield fixed exponents when the corresponding dimensionless parameters are absent. Note that the subdivision into Type A and Type B is defined by the resulting self-similar form rather than by its derivation, which can vary from problem to problem.

\section{Discussion}

\subsection{Classification of self-similar solutions}

As we have formulated self-similarity in terms of scale-invariance, it was found that self-similarity is characterized by what kind of scale transformations it involves. Figure~\ref{fig:F8} summarize all kinds of scale transformations which are formulated in the paper. Then three kinds of self-similar solutions can be distinguished according to which scale transformation the solutions posses. 
\begin{figure*}
\centering
\includegraphics[width=12 cm]{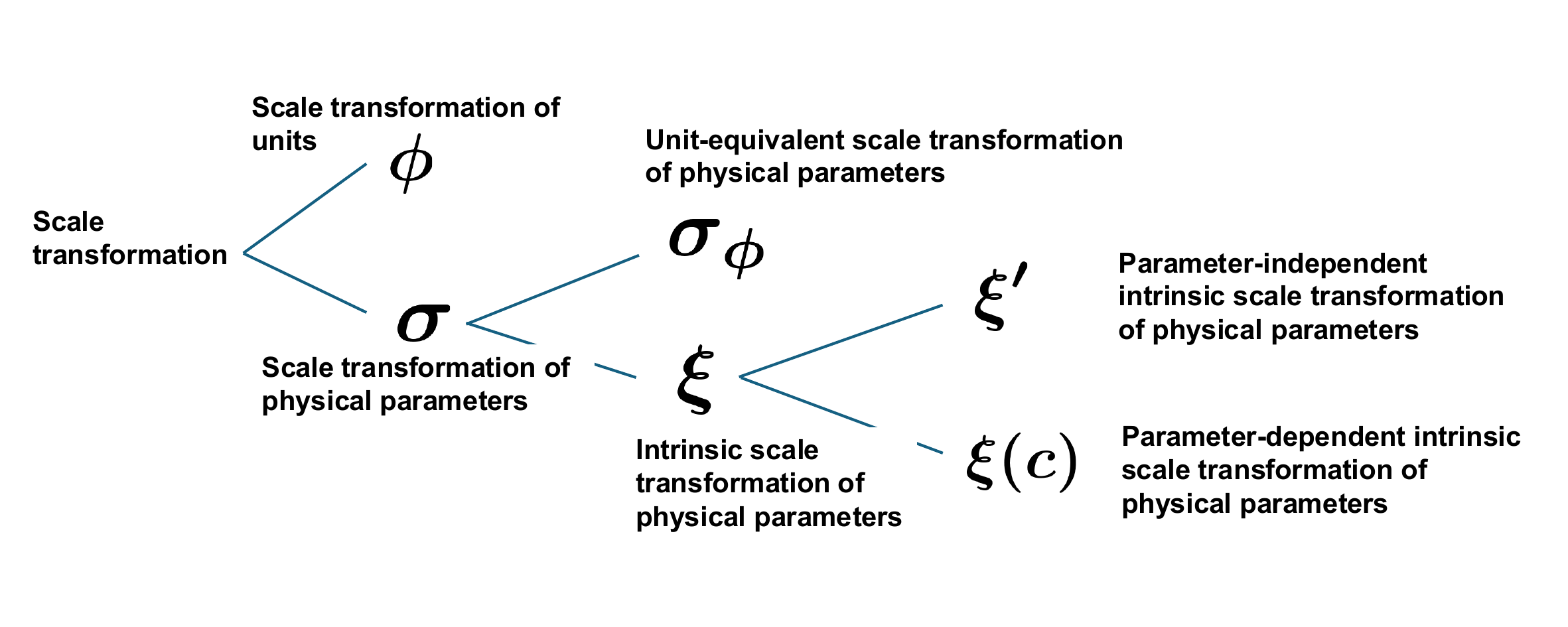} %
\caption{Classification of scale transformations. In every problem involving physical parameters, there are two kinds of scale transformations: scale transformations of units $\boldsymbol{\phi}$ and scale transformations of physical parameters $\boldsymbol{\sigma}$. Scale transformations of physical parameters are further divided into those that are equivalent to scale transformations of units $\boldsymbol{\sigma_\phi}$ and those that are intrinsic $\boldsymbol{\xi}$. Intrinsic scale transformations of physical parameters are further divided into those that are parameter-independent $\boldsymbol{\xi}'$ and those that are parameter-dependent $\boldsymbol{\xi}(c)$. }
\label{fig:F8}
\end{figure*}

Based on definition \ref{dfn:d3} and \ref{dfn:d4}, self-similar solutions can be summarized into three groups, as shown in Table.~\ref{tb:t2}. We can see that the three kinds of self-similar solutions can be clearly distinguished by their scale functions. If a system possesses only the scale functions derived from the structure of units, namely $\boldsymbol{\phi}$ and $\boldsymbol{\sigma_\phi}$, then the system belongs to similarity of the first kind (case I). On the other hand, if a system possesses not only the scale functions derived from units but also intrinsic scale functions associated with physical parameters, namely $\boldsymbol{\phi}$, $\boldsymbol{\sigma_\phi}$, and $\boldsymbol{\xi}$, then it belongs to similarity of the second kind. Within this category, if the exponents of $\boldsymbol{\xi}$ are composed of fixed constant $\boldsymbol{\xi^\prime}$, the system belongs to a self-similar solution of the second kind Type A (case II). If the exponents of $\boldsymbol{\xi}$ are dependent on additional dimensionless parameters and therefore take values from a continuous range $\boldsymbol{\xi(c)}$, the system belongs to a self-similar solution of the second kind Type B (case III).

\begin{table}[h]
\centering
\begin{tabular}{lllllll}
\hline
Kinds of self-similarity & Hierarchy & Subdivision    &\multicolumn{2}{c}{Scale function}& Similarity exponent & Self-similar form \vspace{.5mm}\\  
\hline
   & &  & Unit & \makecell[l]{Physical \\ parameters} &  & 
 \vspace{.5mm}\\
\hline
\makecell[l]{Self-similar solution of \\ the first kind} &  First class & & $\boldsymbol{\phi}$  & $\boldsymbol{\sigma_\phi}$ & Fixed  &  $\Pi = \Phi\left( \zeta\right)$ \vspace{.5mm}\\ 
\multirow{2}{*}{\makecell[l]{Self-similar solution of \\ the second kind}} & \multirow{2}{*}{\makecell[l]{First class \& \\ second class}} & Type A  & $\boldsymbol{\phi}$  & $\boldsymbol{\sigma_\phi}$~$\boldsymbol{\xi^\prime}$ & Fixed &  $\frac{\Pi}{\eta^\alpha} = \Phi\left( \frac{\zeta}{\eta^\beta}\right)$ \vspace{.5mm}\\ 
& &  Type B  & $\boldsymbol{\phi}$  & $\boldsymbol{\sigma_\phi}$~$\boldsymbol{\xi(c)}$ & Parameter -dependent&  $\frac{\Pi}{\eta^{\alpha(c)}} = \Phi\left( \frac{\zeta}{\eta^{\beta(c)}}\right)$ \vspace{.5mm}\\ 
\hline
\end{tabular}
\caption{The table summarizes the three kinds of self-similar solutions. Self-similar solutions of the first and second kinds are distinguished by whether they possess only the first class or both the first and second classes, or equivalently, whether they possess only the basic scale transformations $\boldsymbol{\phi}$ and $\boldsymbol{\sigma_\phi}$ or also the intrinsic scale transformation of physical parameters $\boldsymbol{\xi}$. Self-similar solutions of the second kind are further subdivided into Type A and Type B according to whether their similarity exponents are parameter-independent or parameter-dependent.}
\label{tb:t2}
\end{table}

This subdivision follows the viewpoint of Eggers and Fontelos. Their formulation is based on the transformation of a PDE into an ODE through a self-similar ansatz. Therefore, it may appear specific to problems governed by PDEs. Its applicability to problems for which a PDE is not the fundamental governing equation is less obvious. In this sense, this subdivision relies more on the form of the self-similar solution itself.

The formulation of similarity of the first kind and similarity of the second kind originates from the structure of the numerical-value function $N$ (Eq.~\ref{eq:e25}), which does not distinguish between scale transformations that are physically different operations. No one would regard a scale transformation of units and a scaling of physical parameters as the same physical operation, yet they are mathematically indistinguishable. Since the structure of the scale functions associated with units is known {\it a priori}, we can exploit this information to infer the scale invariance of physical parameters as well. In a sense, the coincidence between the scale invariance of units and that of physical parameters is accidental. Most problems belong to similarity of the second kind, for which these two scale invariances are not equivalent.

In the end, comparing with the formulations of similarity of the first kind and the second kind, by Barenblatt where the kinds of similarity is discriminated through the convergence of the functions \cite{Barenblatt1996,Barenblatt2003}, and the forlumation by Eggers and Fontelos where the kinds of similarity is discriminated by whether the power exponents of similarity variables exists for a particular value or determined by the regularity condition \cite{Eggers2015}, the present formulation in definition.~\ref{dfn:d3} appears to be different from these. Strictly speaking, these two formulation can be reconciled under additional conditions. In the present formulation, the fact that the dimensionless function converges to a finite limit can be partially understood that the scaling function has already exploited its self-similarity. The self-similar form converges to a finite limit due to its scale invariance. However, if there is an intrinsic physical scale transformation $\xi$, the dimensionless function has not exploited to its self-similarity then it diverges in that form while its convergence is finally recovered by introducing the similarity parameters of the second class to exploit all the scale invariance as it is discussed in Appendix~\ref{sec:appendix_C}. However, its convergence must be satisfied by its scaling function $\Phi$. Generally, fully exploited scaling function is not a power function. Therefore, the convergence should be expected for many cases.

As for the formulation by Eggers and Fontelos \cite{Eggers2015}, which is seemingly specialized in PDEs, similarity of the first kind is equivalent in the present formulation. Most of the problems refered as self-similar solution of the second kind by Eggers and Fontelos, Barenblatt are self-similar solution of the second kind Type B as they associate solutions with nonlinear eigenvalue problems, renormalization group theory, singularity and regularity conditions. These solutions involve the breaking conservation laws. The similarity parameters are deduced in the course where this conservation law is recovered. Therefore, the convergence plays an essential role. However, when it is associated with scale invariance and dimensional analysis, the condition of convergence should be embedded into the condition of scaling function.

The distinction between similarity of the first kind and self-similar solutions of the second kind Type A is more subtle. However, the $LMT$ unit system is generally employed. This means that the number of parameters that can be eliminated by dimensional analysis is already fixed, and at most three independent parameters can be removed\footnote{The number of parameters that can be reduced is equivalent to the dimension of the orbit of the scale transformation acting on the parameters. See pp.~86--89 of Ref.~\cite{Olver1986}.}. Consequently, if a problem involves only a small number of governing parameters, it is more likely to reduce to similarity of the first kind. This observation poses a question concerning the choice of unit systems: is the definition of units unique? Is there an alternative system of units under which a larger class of problems would belong to similarity of the first kind? If one adopts a unit system with a higher-dimensional orbit, more parameters could in principle be removed by dimensional analysis. However, considering Case II, the scale factors $A$ and $B$ originate from the underlying physical model, namely the Maxwell viscoelastic foundation model \cite{Maruoka23}. It therefore appears difficult to introduce another universal dimension that would systematically absorb such intrinsic scale transformations. Nevertheless, this remains an open and interesting question.

\subsection{Hierarchy of self-similarity and intermediate asymptotics}

The formulation of self-similarity based on scale invariance naturally leads to the hierarchical structure of self-similarity more obviously. The class is actually originated from the existence of two different scale transformations in physical parameters. Dimensional analysis gives similarity of the first class while similarity parameters of the second class consist of these dimensionless parameters, which are visible in Fig.~\ref{fig:F2a} and Fig.~\ref{fig:F7}. 

This hierarchy is important as this structure highlights the dependence of similarity variables. When all the scale invariance is fully exploited, fully exploited scaling function $\Phi$ describes the dependence of each similarity variables. Suppose there is a following function, $y=f(t,x)$, self-similar form $\Pi= \Phi\left(\eta \right)$ $\Pi =y/t^\alpha$ and $\eta = t^\beta/x$. Suppose that $\Phi$ is fully exploited scaling function. In this formulation, the dependence of similarity variable $\eta$ on $\Pi$ is expressed in $\Phi$. For instance, when $\Phi$ converges to a finite limit, $\Phi \to {\rm const}$ as
$Z \to 0$, it implies that the dependence of $\eta$ is eliminated. As $\Pi = {\rm const}$ then we have the following scaling law in the scale range where this convergence is satisfied, $y= {\rm const} t^\alpha~ (0 <t\ll x^{1/\beta})$, which is called {\it intermeidiate asymptotics}\cite{Barenblatt1972,Barenblatt1996, Barenblatt2003, Barenblatt2014}. This is the case where the dependence of similarity parameters are eliminated and they were expressed as the convergence of scaling function $\Phi$ to a finite limit. 

On the other hand, when the similarity parameter interferes with another similarity parameter, it can lead to crossover of scaling law where the scaling law transforms to another in a certain scale range \cite{Yokota2011, Maruoka23}. This framework of crossover of scaling law is discussed in Ref.~\cite{Maruoka23}, case II. In the problem, there is a self-similar solution of the second type A, $\Psi = \Phi_\mathrm{ex}\left(Z \right)$, Eq.~\ref{eq:e48} and Eq.~\ref{eq:e49b}, where $\Psi$ is CGM number, dervied from the scaling law in the elastic limit \cite{Chastel2016}, and $Z$ is inverse Deborah number, the ratio of the contact time and the relaxation time. In the elastic limit where $Z \to 0$, $\Phi_\mathrm{ex} \to \frac{2}{3}$, then the CGM solution is recovered $ \Pi =\frac{2}{3} \kappa ^{1/3} \eta^{1/3}$. However, in the viscoelastic regime where $Z \sim 1$, $\Psi$ is interfered by $Z$, then the scaling law is changed following the perturbative solution $\Pi = \frac{\kappa}{54 \theta^2} +(\frac{\kappa^2}{486 \theta^3})^{1/3} \eta^{1/6} + \frac{2}{3} \kappa ^{1/3} \eta^{1/3}$. Note that the scaling law of similarltiy variables of first class $\Pi - \eta$ or physical parameters $\delta_m - v_i$ are governed by the similarity variables of the second class $\Psi$ and $Z$. The crossover of scaling law is always caused by the interference of similarity parameter of {\it higher} class. In this case, $Z$ is {\it interfering similarity parameter}  and $\mu$ is {\it driving parameter} \cite{Watanabe25}. However, note that interference of similarity parameters does not always lead to crossover but it can result in the vanishement of self-similarity.

As we have seen, $\Phi_\mathrm{ex}$ actually represents the structure of {\it non}-scale-invariant structure as all the scale invariance are exploited in similarity variables $\Psi$ and $Z$. So there is a benefit to describe the problem by self-similar form as it is easily discriminate the non-scale-invariant structure and scale-invariant structure. This relation is summarized in Fig.~\ref{fig:F9}.
\begin{figure}
\centering
\includegraphics[width=10.5 cm]{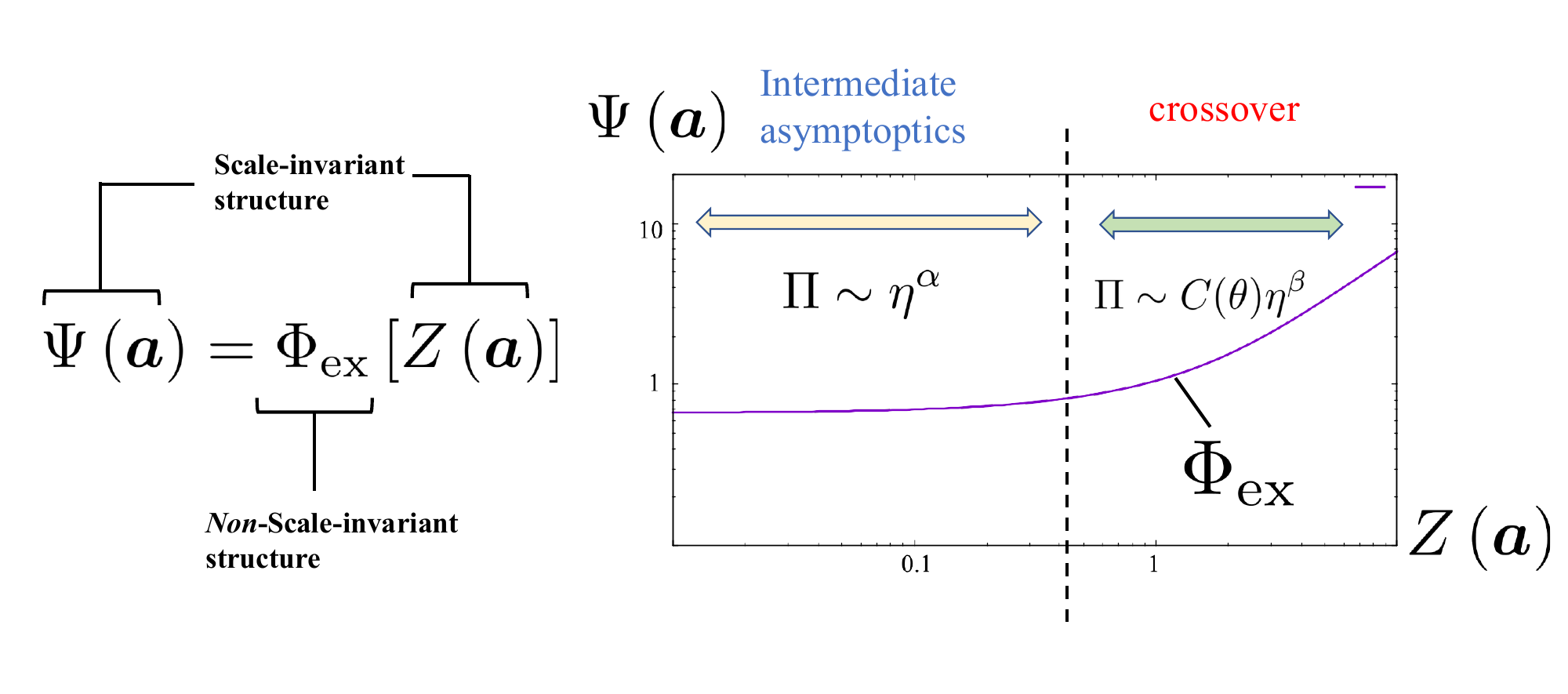}
\caption{The framework of crossover of scaling law. Similarity parameters $\Psi(\boldsymbol{a})$ and $Z(\boldsymbol{a})$ are scale-invariant functions. This functions are invarinat under scale transformation $\boldsymbol{\sigma}$ acting on lower hierarichical parameters $\boldsymbol{a}$ as $\Psi(\boldsymbol{\sigma a})=\Psi(\boldsymbol{a})$  and $Z(\boldsymbol{\sigma a})=Z(\boldsymbol{a})$. On the other hand, the fully exploited scaling function $\Phi_\mathrm{ex}$ are remaining structure on scale-invariant, which reflects {\it non}-scale-invariant structure. This relation is clear when plotting $\Psi$ and $Z$. The curve is described by $\Phi_\mathrm{ex}$. Suppose that $\Pi$, $\eta$ and $\theta$ are lower-hierarchical similarity parameters and $\Psi$ form a scaling law as $\Psi = \Pi/\eta^\alpha$, $Z$ includes a driving parameter $\theta$ respectively. There are the regime where $\Phi_\mathrm{ex}$ converges to a finite limit (intermediate asymptotics) and the regime where there is a functional dependence (crossover).}
\label{fig:F9}
\end{figure}

The present formulation succeeded in making clear the scale-invariant structure to identify non-scale-invariant structure as a function $\Phi_\mathrm{ex}$. It is quite effective to formalize self-similar solutions in terms of scale invariance. However, the appearance of intermediate asymptoitcs and crossover of scaling law depend on fully exploited scaling function $\Phi_\mathrm{ex}$. 


\section{Conclusion \label{sec:SVI}}

In this work, based on scale invariance, we formulated the concepts of exploitation of self-similarity, similarity of the first kind, and similarity of the second kind. By defining scale transformations and their scale functions, we showed that the exploitation of self-similarity can be understood as a transformation into a form invariant under all scale transformations. After formalizing this concept, we introduced the structure of physical parameters, which consist of the product of numerical values and units. This formulation leads to the observation that there are two kinds of scale transformations acting on physical parameters: changes of units and scale transformations of physical parameters themselves. We then formulated dimensional analysis as a transformation of functions into forms invariant under scale transformations of units. Furthermore, we demonstrated that dimensional analysis fully exploits self-similarity only for problems in which the scale invariance of units coincides with that of physical parameters. This insight leads naturally to the definitions of similarity of the first kind and similarity of the second kind according to whether self-similarity is fully exploited by dimensional analysis. We further classified self-similar solutions of the second kind according to whether the exponents of the similarity parameters depend on dimensionless parameters.

The distinction between similarity of the first kind and similarity of the second kind originates from the structure of physical parameters defined in Eq.~\ref{eq:e25}. From this formulation, we find that there are two kinds of scale transformations acting on physical parameters. One acts on the physical parameters themselves, whereas the other acts on their units. This distinction leads to scale invariances of different origins. However, since the structure of scale transformations of units is known {\it a priori}, we can transform functions into forms invariant under scale transformations of units. Owing to the indistinguishability of numerical values, this scale invariance can also be shared by physical parameters. This is the reason why dimensional analysis works for physical problems.

Through this formulation, we find that there are three kinds of self-similar solutions in general. This classification is sufficiently broad that all self-similar physical problems may be assigned to one of these categories. The formulation provides a clearer perspective on self-similar phenomena. It is highly useful to identify the type of self-similarity exhibited by a given problem. We hope that this framework will serve as a useful guide for understanding self-similarity in a wide range of physical systems.

\backmatter

\bmhead{Acknowledgements}

The work was supported by
Grants-in-Aid of MEXT, Japan for Scientific Research, Grant
No. JP24K17022. The author wishes to thank S. Mandre, A. Saxena, and Y. Tagawa for fruitful discussions, as well as the participants of the minicourse Self-Similarity and Scaling for Physicists Working with Numerical Data at OIST. The author is especially grateful to Y. Hirono for many patient and insightful discussions that led to the formulation of physical parameters developed in this paper. 

\bmhead{Author Contribution Statement}

The author confirms sole responsibility for the following: study conception and design, data processing, analysis and interpretation of results, and manuscript preparation.

\bmhead{Data availability}

My manuscript has no associated data.

\begin{appendices}

\section{Proof on the functional form of scale functions} \label{sec:appendix_A}

In this appendix, we prove the following theorem: the function that determines the factor acting on the parameters under scale transformation, is called {\it scale function}, and it is always powerl-law monomials.
\begin{equation}
 \sigma_i\left( A_1,\cdots, A_k \right) =A_1^{\omega_{1i}}\cdots A_k^{\omega_{ki}}. \label{eq:e3}
\end{equation}

\begin{proof}
     
Let us consider the scale transformation with scale factor $A_1, \cdots, A_k$ by $\sigma\left( A_1, \cdots, A_k\right)$ and another scale transformation with scale factor $B_1, \cdots, B_k$ by $\sigma\left( B_1, \cdots, B_k\right)$ as follows,
 \begin{eqnarray}
     a_i^A &=& \sigma_i\left( A_1,\cdots, A_k \right) a_i  \label{eq:e4}\\
     a_i^B &=& \sigma_i\left( B_1,\cdots, B_k \right) a_i \label{eq:e5}.
 \end{eqnarray}
 We assume scale functions $\sigma_i$ are continuously differentiable.

If we consider the scale transformation from $a_i^B$ to $a_i^A$, this can be realized by the following transformation,
\begin{equation}
    a_i^A = \sigma_i\left( \frac{A_1}{B_1},\cdots, \frac{A_k}{B_k} \right) a_i^B.  \label{eq:e6}
\end{equation}
Using Eq.~\ref{eq:e4},\ref{eq:e5},\ref{eq:e6}, we get the following functional equation,
\begin{eqnarray}
    \frac{\sigma_i\left( A_1,\cdots, A_k \right)}{\sigma_i\left( B_1,\cdots, B_k \right)} &=& \sigma_i\left( \frac{A_1}{B_1},\cdots, \frac{A_k}{B_k} \right). \label{eq:e7}  
\end{eqnarray}
Now let us show that the solution of Eq.~\ref{eq:e7} is Eq.\ref{eq:e3}.

Let us differentiate Eq.~\ref{eq:e7} with respect to $A_1$,
\begin{equation}
    \frac{\partial }{\partial A_1} \frac{\sigma_i\left( A_1,\cdots, A_k \right)}{\sigma_i\left( B_1,\cdots, B_k \right)} = \frac{1}{B_1} \frac{\partial}{\partial A_1/B_1} \sigma_i\left( \frac{A_1}{B_1},\cdots, \frac{A_k}{B_k} \right). \nonumber
\end{equation}
Taking $B_1 = A_1, \cdots, B_k = A_k$, we have
\begin{equation}
     \frac{\frac{\partial \sigma_i\left( A_1,\cdots, A_k \right) }{\partial A_1}}{\sigma_i\left( A_1,\cdots, A_k \right)} = \frac{1}{A_1} \frac{\partial}{\partial A_1/B_1} \sigma_i\left( \frac{1}{1},\cdots, \frac{1}{1} \right). \nonumber
\end{equation}
By introducing 
\begin{equation}
      \frac{\partial}{\partial A_1/B_1} \sigma_i\left( 1,\cdots, 1 \right) = \omega_{1i} \nonumber
\end{equation}
then we have the following relation,
\begin{equation}
    \int \frac{\partial \sigma_i\left( A_1,\cdots, A_k \right)}{\sigma_i\left( A_1,\cdots, A_k \right)} = \omega_{1i} \int \frac{\partial A_1}{A_1}. \label{eq:e8}
\end{equation}
The solution of Eq.~\ref{eq:e8} will be 
\begin{equation}
     \sigma_i\left( A_1,\cdots, A_k \right) = A_1^{\omega_{1i}} C \left(A_2, \cdots, A_k \right). \label{eq:e9}
\end{equation}
If we repeat this for $A_2,\cdots, A_k$, we have
\begin{equation}
     \sigma_i\left( A_1,\cdots, A_k \right) = {\rm const}~A_1^{\omega_{1i}}\cdots A_k^{\omega_{ik}}. \label{eq:e10}
\end{equation}
Introducing ${\rm const} =1$, we have
\begin{equation}
     \sigma_i\left( A_1,\cdots, A_k \right) = A_1^{\omega_{1i}} \cdots A_k^{\omega_{ik}}. \nonumber
\end{equation}
which is Eq.~\ref{eq:e3}. \end{proof}

\section{The proof of thorem on functionally independent scale functions}
\label{sec:appendix_B}

Here we prove Theorem 1: A scale transformation is always possible such that any quantity, say $x_1$, in the set of parameters  with functionally independent scale functions $x_1,\cdots, x_k$ changes its numerical value by a specified factor $A$ while the other quantities remain unchanged.

\begin{proof}
    
Let us consider the scale transformation having the orbits $A_1,\cdots, A_k$ acting on $x_1,\cdots,x_k$,
\begin{eqnarray}
        \sigma_1\left(A_1,\cdots,A_k\right) &=& A_1^{\omega_{11}}, \cdots, A_k^{\omega_{k1}},\nonumber \\
        \vdots  \nonumber \\ 
        \sigma_{k}\left( A_1,\cdots,A_k\right) &=& A_1^{\omega_{11}}\cdots A_k^{\omega_{kk}}.
\end{eqnarray}
    Then we obtain as follows,
\begin{eqnarray}
        x_1^{'} &=& x_1 A_1^{\omega_{11}}\cdots A_k^{\omega_{k1}} \nonumber \\ 
        x_2^{'} &=& x_2 A_1^{\omega_{12}}\cdots A_k^{\omega_{k2}} \nonumber \\ 
        \vdots \nonumber \\
        x_k^{'} &=& x_k A_1^{\omega_{1k}}\cdots A_k^{\omega_{kk}}. \label{eq:e12}
\end{eqnarray}

Now let us the transformation as it is illustrated, such that any quantity, say $x_1$, in the set of parameters  with functionally independent scale functions $x_1,\cdots, x_k$ changes its numerical value by a specified factor $A$ while the other quantities remain unchanged,
\begin{equation}
    x_1^{'} =A x,~x_2^{'} = x_2,\cdots, x_k^{'} = x_k. \label{eq:e13}
\end{equation}
Therefore, we have,
\begin{eqnarray}
     A_1^{\omega_{11}}\cdots A_k^{\omega_{k1}}=A,~A_1^{\omega_{12}}\cdots A_k^{\omega_{k2}}  = 1,\cdots, \nonumber \\ 
     A_1^{\omega_{1k}}\cdots A_k^{\omega_{kk}} = 1. \label{eq:e14}
\end{eqnarray}
By taking logarithms, we obtain a system of linear equations,
\begin{eqnarray}
    \omega_{11} \ln A_1+ \omega_{21} \ln A_2 \cdots \omega_{k1} \ln A_k &=& \ln A \nonumber \\ 
    \omega_{12} \ln A_1+ \omega_{22} \ln A_2 \cdots \omega_{k2} \ln A_k &=& 0  \nonumber \\ 
    \vdots \nonumber \\
    \omega_{1k} \ln A_1+ \omega_{2k} \ln A_2 \cdots \omega_{kk} \ln A_k &=& 0. \label{eq:e15}
\end{eqnarray}
The system of Eqs.~\ref{eq:e15} has at least one solution. Indeed,  it is insoluble only if the left-hand side of the first equation is a linear combination of the left-hand sides of the remaining equations,
\begin{eqnarray}
    \omega_{11} \ln A_1+ \omega_{21} \ln A_2 \cdots \omega_{k1} \ln A_k \nonumber \\
    = c_2 \left( \omega_{12} \ln A_1+ \omega_{22} \ln A_2 \cdots \omega_{k2} \ln A_k \right) \nonumber \\
    + \cdots + c_k\left(  \omega_{1k} \ln A_1+ \omega_{2k} \ln A_2 \cdots \omega_{kk} \ln A_k \right). \label{eq:e16}
\end{eqnarray}
Therefore, we have
\begin{equation}
    \sigma_1\left(A_1,\cdots,A_k\right) = \sigma_2\left(A_1,\cdots,A_k\right)^{c_2} \cdots \sigma_k\left(A_1,\cdots,A_k\right)^{c_k}. \label{eq:e17} 
\end{equation}
However, Eq.~\ref{eq:e17} is inconsistent with the fact that $\sigma_1,\cdots, \sigma_k$ are functionally independent as they are expressed by a scale functon of the others. Therefore, the theorem was proved. 
\end{proof}

\section{The mathematical condition for the existence of correspondence between two scale functions}\label{sec:appendix_Ca}

Suppose that ${\bf p}$ and ${\bf q}$ are vectors of power exponents of two similarity parameters. There is a vector ${\bf c}$ such that are orthogonal to ${\bf p}$ and ${\bf q}$,  ${\bf cp } = {\bf cq }= 0$. Then the scale invariance of a system is characterized by the following vector space,  
\begin{equation}
\begin{aligned}    
{\mathcal C} 
&\coloneqq 
{\rm span}
\{ {\bf c}^{(\alpha)} \}_{\alpha = 1, \ldots, p -2} 
\\
&= 
({\rm span}\{\bf p\})^\perp
\cap 
({\rm span}\{\bf q\})^\perp
\\
&= 
({\rm span}\{\bf p, \bf q \})^\perp,
\end{aligned} \label{eq:eac1}
\end{equation}
where  $\bf V^\perp$ denotes the orthogonal complement of a vector space $\bf V$ \cite{Watanabe25}.

For example, the scale functions of $\xi$ in case II was described in Eq.~\ref{eq:e47} and their similarity variables invariant under the transformation $\xi$ are Eq.~\ref{eq:e48}. The scale invariance of $\xi$ is characterized by the following vector space,
\begin{equation}
    c_1\left[p_\Pi,p_\kappa,p_\eta,p_\theta \right]_\Psi + c_2 \left[q_\Pi,q_\kappa,q_\eta,q_\theta \right]_Z = c_1\left[3,-1,-1,0 \right]_{\Psi}+c_2\left[1,0,-\frac{1}{2}, -1 \right]_Z, \label{eq:eac2} 
\end{equation}
where $\left[p_\Pi,p_\kappa,p_\eta,p_\theta \right]_\Psi$ and $\left[q_\Pi,q_\kappa,q_\eta,q_\theta \right]_Z$ are power exponents of $\Pi, \kappa, \eta, \theta$ of $\Psi$ and $Z$ respectively. Figure \ref{fig:FC1} shows the data collapse using different choices of similarity parameters of the second class. Note that all the similarity parameters satisfy the condition of Eq.~\ref{eq:eac2} and all similarity parameters reveal data collapse. This suggest that all the parameters in Fig.~\ref{fig:FC1} are equivalent in terms of their vector space. This point is important. There are infinite variations of the compositions of similarity variables.  
\begin{figure*}
\centering
\includegraphics[width=12 cm]{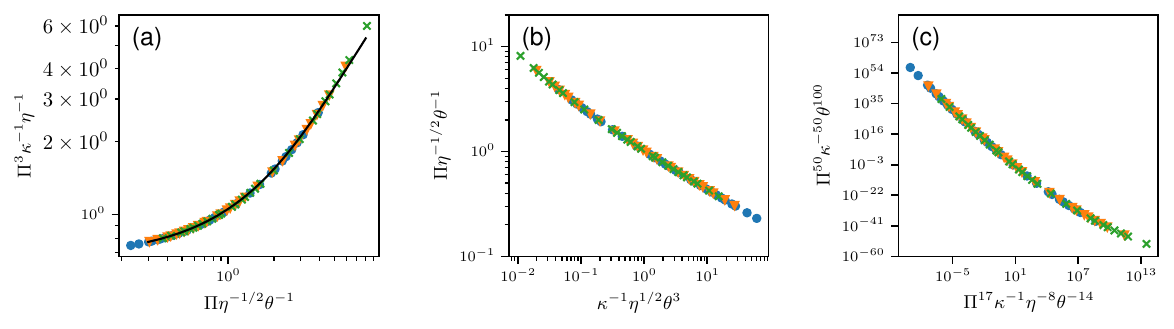} %
\caption{The plots using different choice of similarity parameters of the second class, $\Pi^3 \kappa^{-1} \eta^{-1}$ ($c_1=1, c_2=0$) vs $\Pi\eta^{-1/2}\theta^{-1}$  ($c_1=0, c_2=1$) (a), $\Pi \eta^{-1/2}\theta^{-1}$ ($c_1=0, c_2=1$) vs $\kappa^{-1} \eta^{1/2} \theta^3$ ($c_1=1, c_2=-3$) (b), $\Pi^{50}\kappa^{-50}\theta^{100}$ ( $c_1=50, c_2=-100$) vs $\Pi^{17}\kappa^{-1}\eta^{-8}\theta^{-14}$ ($c_1=1, c_2=14$) (c). Note that all the variation of similarity parameters are obeying the column vectors spanned by ${\bf p}$ and ${\bf q}$ as Eq.~\ref{eq:eac2}.
}
\label{fig:FC1}
\end{figure*}

\section{The relation with the definitions by Barenblatt }\label{sec:appendix_C}

According to Barenblatt's formulation, similarity of the second kind are defined as the function of similarity parameters of the first class $\Phi\left(\zeta, \eta \right)$ does not converge to a finite limit as $\eta$ goes to 0 or infinity but the its convergence is recovered by introducing similarity parameters of the second class as $\Phi\left( \zeta/ \eta^\beta \right)/\eta^\alpha$. We will demonstrate how this formulation is related to our formulation. In conclusion, the present formulations based on scale invariance reveals the same results for the scale transformation. However, scale invariance alone does not guarantee convergence in the general limits of the parameters, since such convergence depends on the limiting behavior of the fully exploited scaling function.

Let us start from the exploited self-similar form in Eq.~\ref{eq:e67}. Similarity parameter $Z_1,\cdots,Z_p$ in are functionally dependent, therefore we rewrite the function as 
\begin{eqnarray}
    \frac{\Xi_1}{\Pi_1^{\kappa_{11}} \cdots \Pi_l^{\kappa_{l1}} } = \Phi_\mathrm{ex}\left( \frac{\Xi_2}{\Pi_1^{\kappa_{12}} \cdots \Pi_l^{\kappa_{l2}} },\cdots, \frac{\Xi_p}{\Pi_1^{\kappa_{1p}} \cdots \Pi_l^{\kappa_{lp}} } \right) \label{eq:eC1}
\end{eqnarray}
If we rewrite the function having similarity parameters of the first class, this will be
\begin{eqnarray}
    \Xi_1 =F \left(\Pi_1,\cdots,\Pi_l,\Xi_2, \cdots,\Xi_p \right)= \nonumber \\
    \Pi_1^{\kappa_{11}} \cdots \Pi_l^{\kappa_{l1}} \Phi_\mathrm{ex}\left( \frac{\Xi_2}{\Pi_1^{\kappa_{12}} \cdots \Pi_l^{\kappa_{l2}} },\cdots, \frac{\Xi_p}{\Pi_1^{\kappa_{1p}} \cdots \Pi_l^{\kappa_{lp}} } \right).\label{eq:eC2}
\end{eqnarray}
As a scale transformation acts on Eq.~\ref{eq:eC2}, then we have
\begin{equation}
   \left( \varsigma_1 \Pi_1\right)^{\kappa_{11}} \cdots \left( \varsigma_l \Pi_l\right)^{\kappa_{l1}} \Phi_\mathrm{ex}\left( \frac{\xi_2 \Xi_2}{\left( \varsigma_1 \Pi_1\right)^{\kappa_{12}} \cdots \left( \varsigma_l \Pi_l\right)^{\kappa_{l2}} },\cdots, \frac{\xi_p\Xi_p}{\left( \varsigma_1 \Pi_1\right)^{\kappa_{1p}} \cdots\left(\varsigma_l \Pi_l\right)^{\kappa_{lp}} } \right). \label{eq:eC3}
\end{equation}
The limit of $\Xi_1, \cdots, \Pi_l \to 0~{\rm or}~\infty$ corresponds to $\varsigma_1,\cdots,\varsigma_l \to 0~{\rm or}~\infty$, thus we obtain
\begin{equation}
\begin{aligned}
\left( \varsigma_1 \Pi_1\right)^{\kappa_{11}} \cdots \left( \varsigma_l \Pi_l\right)^{\kappa_{l1}} \Phi_\mathrm{ex}\left( \frac{\xi_2 \Xi_2}{\left( \varsigma_1 \Pi_1\right)^{\kappa_{12}} \cdots \left( \varsigma_l \Pi_l\right)^{\kappa_{l2}} },\cdots, \frac{\xi_p\Xi_p}{\left( \varsigma_1 \Pi_1\right)^{\kappa_{1p}} \cdots\left(\varsigma_l \Pi_l\right)^{\kappa_{lp}} } \right) \xrightarrow[\varsigma_1,\cdots,\varsigma_l \to 0~{\rm or}~\infty]{} 0~{\rm or}~\infty. \label{eq:eC4}
\end{aligned}
\end{equation}
As the function of similarity of the first class Eq.~\ref{eq:eC4} goes to 0 of infinity in the limit of scale factors goes to 0 or infinity. 
However, as for the function of the second class, it will be 
\begin{equation}
\begin{aligned}
   \frac{ F \left(\varsigma_1\Pi_1,\cdots,\varsigma_l \Pi_l,\xi_2 \Xi_2, \cdots,\xi_p \Xi_p \right)}{\left( \varsigma_1 \Pi_1\right)^{\kappa_{11}} \cdots \left( \varsigma_l \Pi_l\right)^{\kappa_{l1}}} = \Phi_\mathrm{ex}\left( \frac{\xi_2 \Xi_2}{\left( \varsigma_1 \Pi_1\right)^{\kappa_{12}} \cdots \left( \varsigma_l \Pi_l\right)^{\kappa_{l2}} },\cdots, \frac{\xi_p\Xi_p}{\left( \varsigma_1 \Pi_1\right)^{\kappa_{1p}} \cdots\left(\varsigma_l \Pi_l\right)^{\kappa_{lp}} } \right) \nonumber \\ = \Phi_\mathrm{ex}\left( \frac{\Xi_2}{\Pi_1^{\kappa_{12}} \cdots \Pi_l^{\kappa_{l2}} },\cdots, \frac{\Xi_p}{\Pi_1^{\kappa_{1p}} \cdots \Pi_l^{\kappa_{lp}} } \right)\xrightarrow[\varsigma_1,\cdots,\varsigma_l \to 0~{\rm or}~\infty]{} {\rm const}. \label{eq:eC5}
\end{aligned}
\end{equation}
As Eq.~\ref{eq:eC5} shows, the convergence of the function is recovered by introducing the similarity parameters of the second class. However, note that the convergence of the general limit of parameters depends on $\Phi_\mathrm{ex}$.

\end{appendices}

\bibliography{Ref} 

\end{document}